\documentclass[9pt]{IEEEtran}
\usepackage{algorithmic}
\usepackage{algorithm}
\usepackage{subfigure}
\usepackage{graphicx}
\usepackage{cite}
\usepackage{threeparttable}
\usepackage{ntheorem}
\usepackage{color}
\usepackage{multirow}
\usepackage{amsmath}
\usepackage{comment}

\newtheorem{problem}{\textbf{Problem}}
\newtheorem{mydefinition}{\textbf{Definition}}
\newtheorem{mytheorem}{\textbf{Theorem}}
\newtheorem{mylemma}{\textbf{Lemma}}

\graphicspath{{./figs/}}

\begin{document}

\title{\Huge
A High-Performance Triple Patterning Layout Decomposer\\ with Balanced Density
%\vspace{-.1in}
}

\iftrue
\author{
\IEEEauthorblockN{
Bei Yu\IEEEauthorrefmark{0},\ \ 
Yen-Hung Lin\IEEEauthorrefmark{2},\ \ 
Gerard Luk-Pat\IEEEauthorrefmark{3},\ \ 
Duo Ding\IEEEauthorrefmark{4},\ \ 
Kevin Lucas\IEEEauthorrefmark{3},\ \ 
David Z. Pan\IEEEauthorrefmark{0}}\\
\IEEEauthorblockA{\IEEEauthorrefmark{0} ECE Dept., University of Texas at Austin, Austin, USA \ \ \ 
\IEEEauthorrefmark{3} Synopsys Inc., Austin, USA}\\
\IEEEauthorblockA{\IEEEauthorrefmark{2} CS Dept., National Chiao Tung University, Taiwan \ \ \ \ \ \ \ 
\IEEEauthorrefmark{4} Oracle Corp., Austin, USA}
\vspace{-.1in}
}
\fi

\maketitle
\thispagestyle{empty} % disable page number

\begin{abstract}
%Triple patterning lithography (TPL) has received more and more attentions from industry as one of the candidates for next generation lithography.
Triple patterning lithography (TPL) has received more and more attentions from industry as one of the leading candidate for 14nm/11nm nodes.
In this paper, we propose a high performance layout decomposer for TPL.
Density balancing is seamlessly integrated into all key steps in our TPL layout decomposition, including
density-balanced semi-definite programming (SDP),  density-based mapping, and density-balanced graph simplification.
Our new TPL decomposer can obtain high performance even compared to previous state-of-the-art layout decomposers which are not balanced-density aware,
e.g., by Yu et al. (ICCAD'11), Fang et al. (DAC'12), and Kuang et al. (DAC'13).
%Specifically, a 65-81\% decrease in the weighted sum of the conflicts and stitches can be achieved.
Furthermore, the balanced-density version of our decomposer can provide more balanced density which leads to less edge placement error (EPE), while the conflict and stitch numbers are still very comparable to our non-balanced-density baseline.
\end{abstract}

\iffalse
\category{EDA9.2}{Design for Manufacturability}{lithography-related design optimizations}
\terms{Algorithms, Performance, Theory}
 % NOT required for Proceedings
\fi

\vspace{-.1in}
\section{Introduction}

As the minimum feature size further decreases, the semiconductor industry faces great challenge in patterning sub-22nm half-pitch due to the delay of viable next generation lithography,
such as extreme ultra violet (EUV) %\cite{EUV_SPIE2010_Arisawa}
and electric beam lithography (EBL). %\cite{EBL_TCAD2012_Yuan}
Triple patterning lithography (TPL), along with self-aligned double patterning (SADP), are solution candidates for the 14nm logic node \cite{ITRS}.
Both TPL and SADP are similar to double patterning lithography (DPL), but with different or more exposure/etching processes \cite{LITH_ICCAD2012_Yu}.
SADP may be significantly restrictive on design, i.e., cannot handle irregular arrangements of contacts and does not allow stitching.
%\cite{DPL_SPIE08_Bencher}
%\cite{SADP_ISQED2011_Mirsaeedi}
%Although the total cost of TPL may be larger, it could be reduced by the faster multi-beam mask writing tools being developed.
Therefore, TPL began to receive more attention from industry, especially for metal 1 layer patterns.
For example, industry has already explored test-chip patterns with triple patterning and even quadruple patterning \cite{2009Intel}.

Similar to DPL, the key challenge of TPL lies in the decomposition process where the original layout is divided into three masks.
During decomposition, when the distance between any two features is less than the minimum coloring distance $dis_m$,
they need to be assigned into different masks to avoid a conflict.
Sometimes, a conflict can be resolved by splitting a pattern into two touching parts, called \textit{stitches}.
After the TPL layout decomposition, the features are assigned into three masks (colors) to remove all conflicts.
The advantage of TPL is that the effective pitch can be tripled which can further improve lithography resolution.
Besides, some native conflicts in DPL can be resolved.

In layout decomposition, especially for TPL, density balance should also be considered, along with the conflict and stitch minimization.
A good pattern density balance is also expected to be a consideration in mask CD and registration control \cite{TPL_SPIE2012_Lucas},
while unbalanced density would cause lithography hotspots as well as lowered CD uniformity due to irregular pitches \cite{DPL_ASPDAC2010_Yang}.
%Besides, the CD and overlay control for each exposure/etching step will be critical in order to ensure that the spacing between features patterned is sufficient.
%Achieving good pattern density balance across the three lithography steps is important for fine etch control.
%TPL also requires that three expensive masks must be manufactured for each device layer, each with significantly improved mask CD control and mask pattern registration.
However, from the algorithmic perspective, achieving a balanced density in TPL could be harder than that in DPL.
(1) In DPL, two colors can be more implicitly balanced; while in TPL, often times existing/previous strategies may try to do DPL first, and then do some ``patch" with the third mask, which causes a big challenge to ``explicitly" consider the density balance.
(2) Due to the one more color, the solution space is much larger \cite{TPL_SPIE08_Cork}.
(3) Instead of global density balance, local density balance should be considered to reduce the potential hotspots, since neighboring patterns are one of the main sources of hotspots.
As shown in Fig. \ref{fig:balance} (a)(b), when only global density balance is considered, feature $a$ is assigned white color.
Since two black features are close to each other, hotspot may be introduced.
To consider the local density balance, the layout is partitioned into four bins \{$b_1, b_2, b_3, b_4$\} (see Fig. \ref{fig:balance} (c)).
Feature $a$ is covered by bins $b_1$ and $b_2$, therefore it is colored as blue to maintain the local density balances for both bins (see Fig. \ref{fig:balance} (d)).

\begin{figure}[tb]
    \centering
    \vspace{-.0in}
    \subfigure[]{\includegraphics[width=0.22\textwidth]{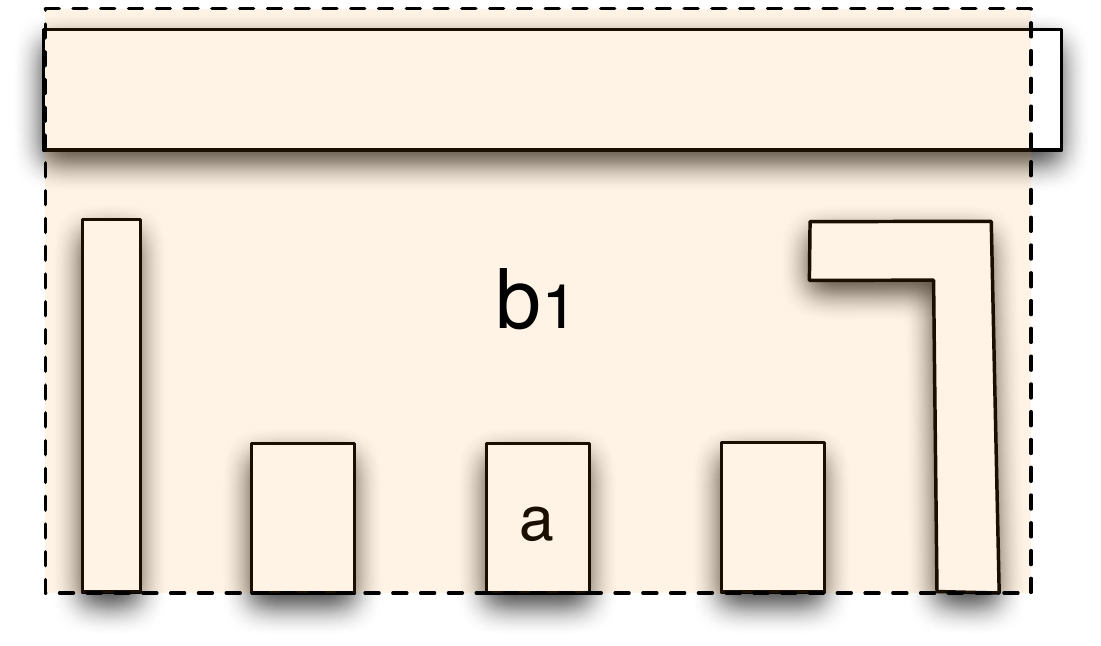}}
    \subfigure[]{\includegraphics[width=0.22\textwidth]{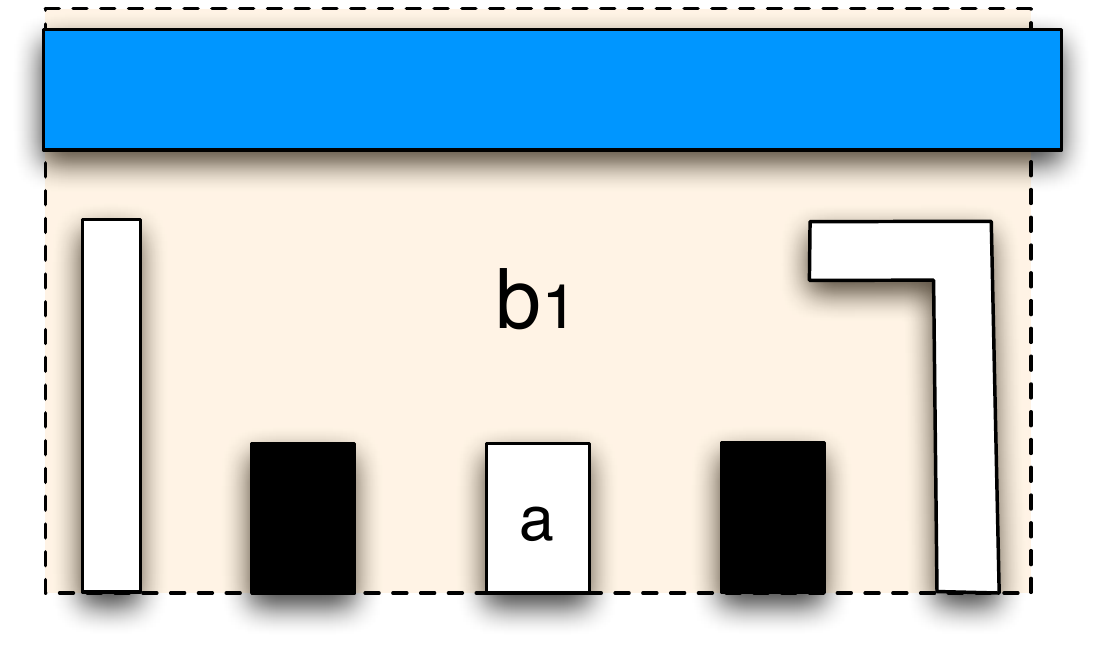}}
    \subfigure[]{\includegraphics[width=0.22\textwidth]{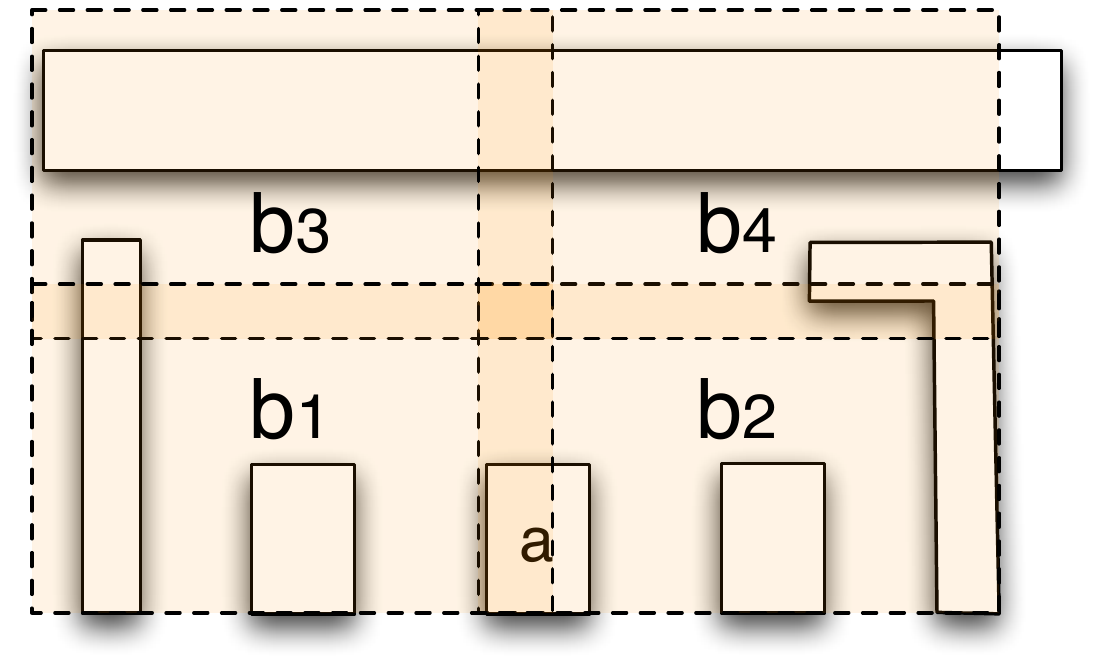}}
    \subfigure[]{\includegraphics[width=0.22\textwidth]{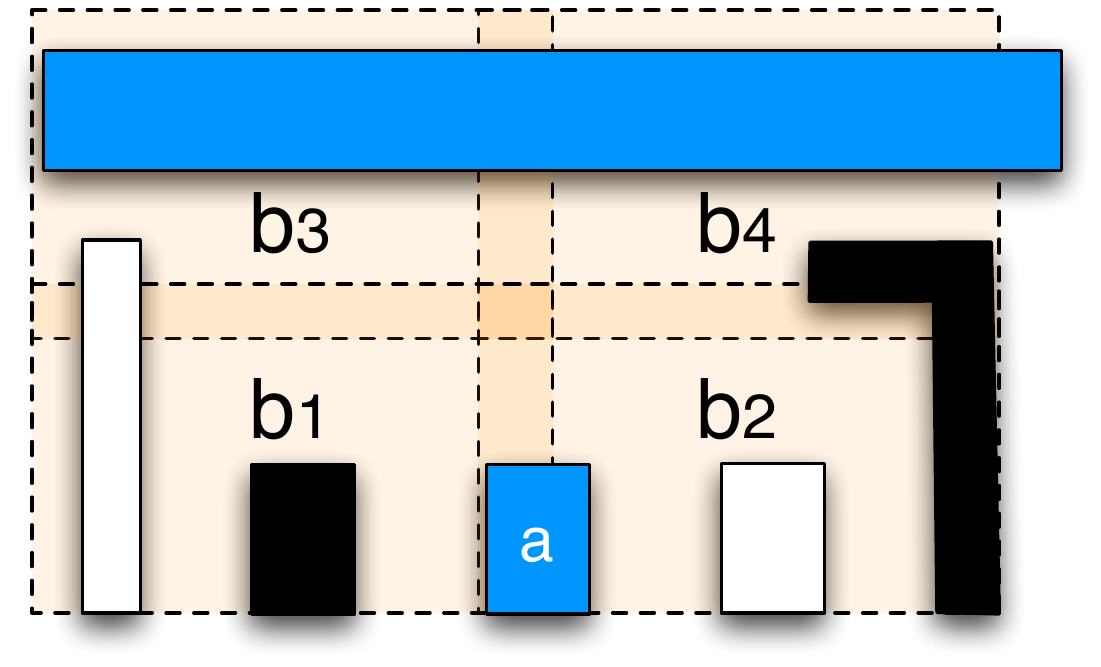}}
    \vspace{-.1in}
    \caption{
    ~Decomposed layout with (a) (b) global balanced density.~(c) (d) local balanced density in all bins.
    }
    \label{fig:balance}
    \vspace{-.2in}
\end{figure}

There are investigations on TPL layout decomposition
\cite{TPL_SPIE08_Cork,TPL_ICCAD2011_Yu,TPL_DAC2012_Fang,TPL_ICCAD2012_Tian,TPLEC_SPIE2013_Yu,TPL_DAC2013_Kuang}
or TPL aware design \cite{DFM_DAC2012_Ma,DFM_ICCAD2012_Lin,DFM_ICCAD2013_Yu}.
\cite{TPL_SPIE08_Cork} provided a three coloring algorithm, which adopts a SAT Solver.
Yu et al. \cite{TPL_ICCAD2011_Yu} proposed a systematic study for the TPL layout decomposition, where they showed that this problem is NP-hard.
%Instead of expensive ILP, they proposed a semidefinite programming (SDP) based approximation to achieve tradeoffs between runtime and solution quality.
Fang et al. \cite{TPL_DAC2012_Fang} presented several graph simplification techniques to reduce the problem size,
and a maximum independent set (MIS) based heuristic for the layout decomposition.
\cite{TPL_ICCAD2012_Tian} proposed a layout decomposer for row structure layout.
%\cite{TPL_SPIE2011_Ghaida} provided a methodology to reuse the decomposer for DPL.
However, these existing studies suffer from one or more of the following issues:
(1) cannot integrate the stitch minimization for the general layout, or can only deal with stitch minimization as a post-process;
(2) directly extend the methodologies from DPL, which loses the global view for TPL;
(3) assigning colors one by one prohibits the ability for density balance.
%(3) density balance is not considered.
%only stitch locations for DPL are considered as stitch candidates.
%This limits the ability of stitches to resolve the triple patterning conflicts and may result in some conflicts that could have been resolved otherwise.

In this paper, we propose a high performance layout decomposer for TPL.
Compared with previous works, our decomposer provides not only less conflict and stitch number, but also more balanced density.
%To our best knowledge, it is the first systematic study for balanced density optimization in triple patterning.
%Note that in this work we do not explicitly consider errors due to overlay error, as it mainly a function of scanner tool control and mask write control. 
%However, it is known that mask write overlay control generally benefits from improved density balance.
In this work, we focus on the coloring algorithms and leave other layout related optimizations to post-coloring stages,
such as compensation for various mask overlay errors introduced by scanner and mask write control processes.
However, we do explicitly consider balancing density during coloring, since it is known that mask write overlay control generally benefits from improved density balance.

Our key contributions include the following.
(1) Accurately integrate density balance into the mathematical formulation;
(2) Develop a three-way partition based mapping, which not only achieves less conflicts, but also more balanced density;
%(3) Suggest a procedure to generate appropriate stitch candidates for TPL, which overcomes the previous limitations;
(3) Propose several techniques to speedup the layout decomposition;
(4) Our experiments show the best results in solution quality while maintaining better balanced density (i.e., less EPE).

The rest of the paper is organized as follows.
Section \ref{sec:problem} presents the basic concepts and the problem formulation.
Section \ref{sec:overview} gives the overall decomposition flow.
Section \ref{sec:algo} presents the details to improve balance density and decomposition performance,
and Section \ref{sec:speedup} shows how we further speedup our decomposer.
Section \ref{sec:result} presents our experimental results, followed by a conclusion in Section \ref{sec:conclusion}.

\vspace{-.1in}
\section{Problem Formulation}
\label{sec:problem}

Given input layout which is specified by features in polygonal shapes, we partition the layout into $n$ bins $B = \{b_1, \dots, b_n\}$.
Note that neighboring bins may share some overlapping.
For each polygonal feature $r_i$, we denote its area as $den_i$, and its area covered by bin $b_k$ as $den_{ki}$.
Clearly $den_i \ge den_{ki}$ for any bin $b_k$.
During layout decomposition, all polygonal features are divided into three masks.
%Similarly, for each bin $b_k$ we define three densities $d_{k1}, d_{k2}$ and $d_{k3}$ as sum of all its feature densities.
%Similarly, we define the density of each bin as the sum of all its feature densities.
For each bin $b_k$, we define three densities ($d_{k1}, d_{k2}, d_{k3}$), where
$d_{kc} = \sum den_{ki}$, for any feature $r_i$ assigned to color $c$.
Therefore, we can define the local density uniformity as follows:
%\begin{define}[Density Uniformity]
%Given three densities $d_1, d_2$ and $d_3$ on three masks, the density uniformity is used to measure the difference of the densities, i.e. $\textrm{max} \{d_i\} / \textrm{min} \{d_i\}$.
%The density uniformity is denoted by $DU$.
%\end{define}

\begin{mydefinition}[Local Density Uniformity]
For the bin $b_k$ $\in S$, the local density uniformity is $\textrm{max} \{d_{kc}\} / \textrm{min} \{d_{kc}\}$ given three densities $d_{k1}, d_{k2}$ and $d_{k3}$ for three masks and is used to measure the ratio difference of the densities.
A lower value means better local density balance.
The local density uniformity is denoted by $DU_k$.
\end{mydefinition}

For convenience, we use the term density uniformity to refer to local density uniformity in the rest of this paper.
It is easy to see that $DU_k$ is always larger than or equal to 1.
To keep a more balanced density in bin $b_k$, we expect $DU_k$ as small as possible, i.e., close to 1.

\begin{problem}[Density Balanced Layout Decomposition]
Given a layout which is specified by features in polygonal shapes, the layout graphs and the decomposition graphs are constructed.
%Note that some rectangles may further be sliced into smaller pieces to resolve conflicts.
Our goal is to assign all vertices in the decomposition graph into three colors (masks) to minimize the stitch number and the conflict number,
while keeping all density uniformities $DU_k$ as small as possible.
\end{problem}

%Given a decomposition graph, we denote the problem of assigning all vertices into three colors as a color assignment problem.
%In other words, density balanced layout decomposition problem incorporates two parts:
%decomposition graph construction and color assignment on each decomposition graph.

\vspace{-.1in}
\section{Overall Decomposition Flow}
\label{sec:overview}

\begin{figure}[htb]
\vspace{-.05in}
  \centering
  \includegraphics[width=0.5\textwidth]{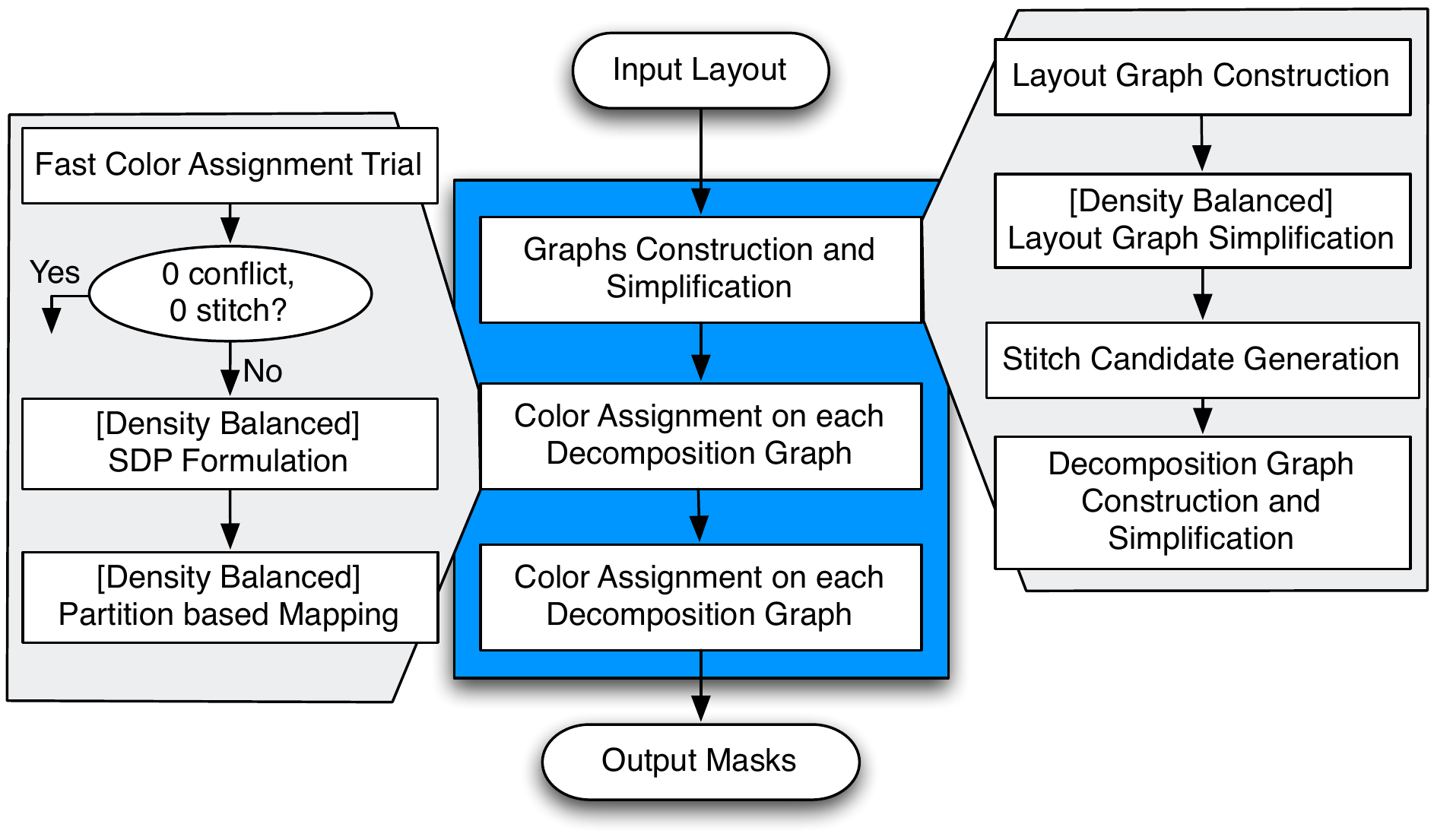}
  \vspace{-.3in}
  \caption{~Overall flow of proposed density balanced decomposer.}
  \label{fig:overview}
 % \vspace{-.1in}
\end{figure}

The overall flow of our TPL decomposer is illustrated in Fig. \ref{fig:overview}.
It consists of two stages: graph construction / simplification, and color assignment.
Given input layout, layout graphs and decomposition graphs are constructed, then graph simplifications \cite{TPL_ICCAD2011_Yu}\cite{TPL_DAC2012_Fang} are applied to reduce the problem size.
Two additional graph simplification techniques are introduced in Sec. \ref{sec:nostitch} and \ref{sec:cluster}.
%Stitch candidate generation is discussed in Appendix \ref{sec:app_stitch}.
During stitch candidate generation, the methods described in \cite{TPL_DAC2013_Kuang} are applied to search all stitch candidates for TPL.
In second stage, for each decomposition graph, color assignment is proposed to assign each vertex one color.
Before calling SDP formulation, fast color assignment trial is proposed to achieve better speedup (see Section \ref{sec:fastassign}).
%It is similar to that in \cite{TPL_ICCAD2011_Yu}, but with some improvements as follows,
%(1) new stitch candidate generation after the layout graphs simplification;
%(2) in decomposition graph simplification, not only detects the bridge edge, but also detects the bridge vertex, which is similar to 2-connected component in \cite{TPL_DAC2012_Fang}.
%(3) fast color assignment trial is proposed to achieve better speed-up.
%(4) new mapping algorithm, which is based on maximum-cut partitioning.
%First it constructs layout graphs from the input layout, and applies layout graph simplification \cite{TPL_ICCAD2011_Yu} which reduces the graph size without losing the optimality.
%Then it carries out stitch candidate generation for TPL.
%Based on the stitch candidates and the layout graphs, it constructs decomposition graphs and apply decomposition graph simplification \cite{TPL_ICCAD2011_Yu}.
%For each decomposition graph it produces the color assignment, which includes the density balanced semidefinite programming (SDP) formulation and density balanced mapping.

\begin{figure}[tb]
\vspace{-.1in}
  \centering
  \subfigure[] {\includegraphics[width=0.15\textwidth]{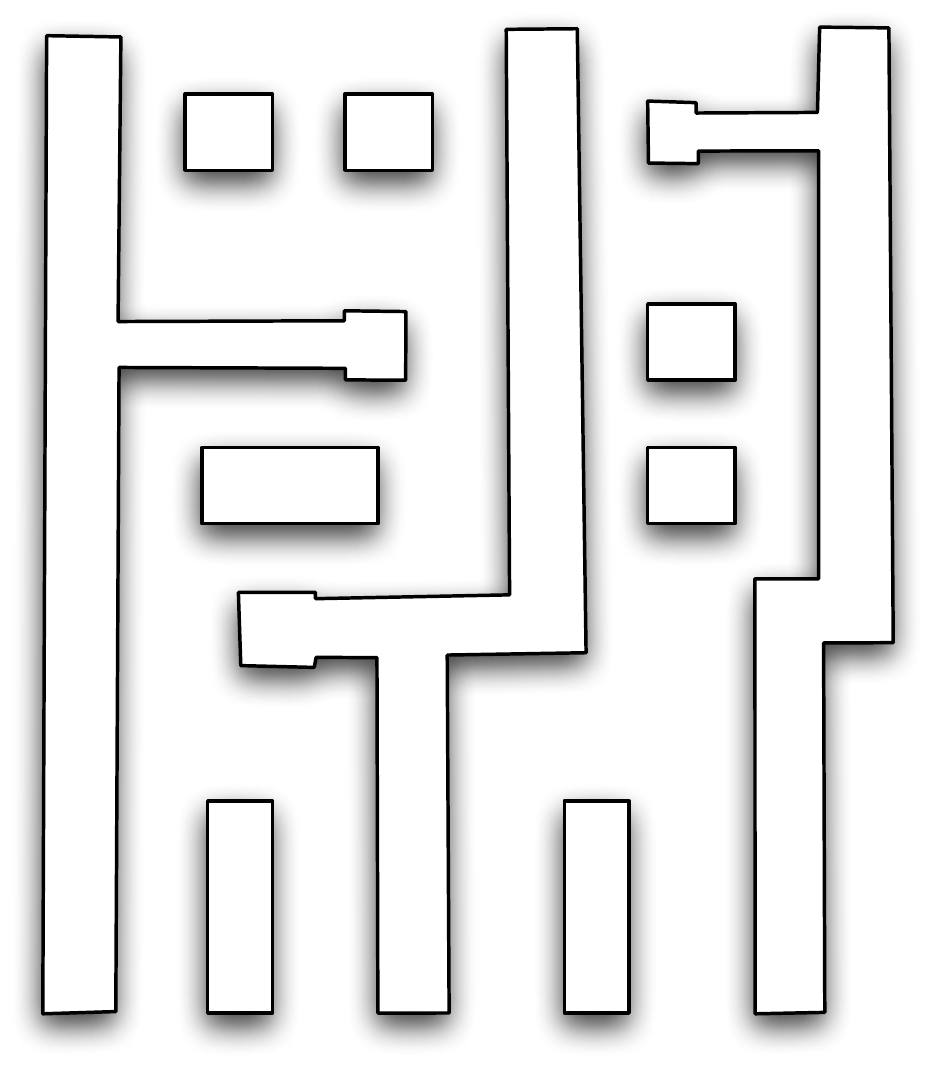}}
  \subfigure[] {\includegraphics[width=0.16\textwidth]{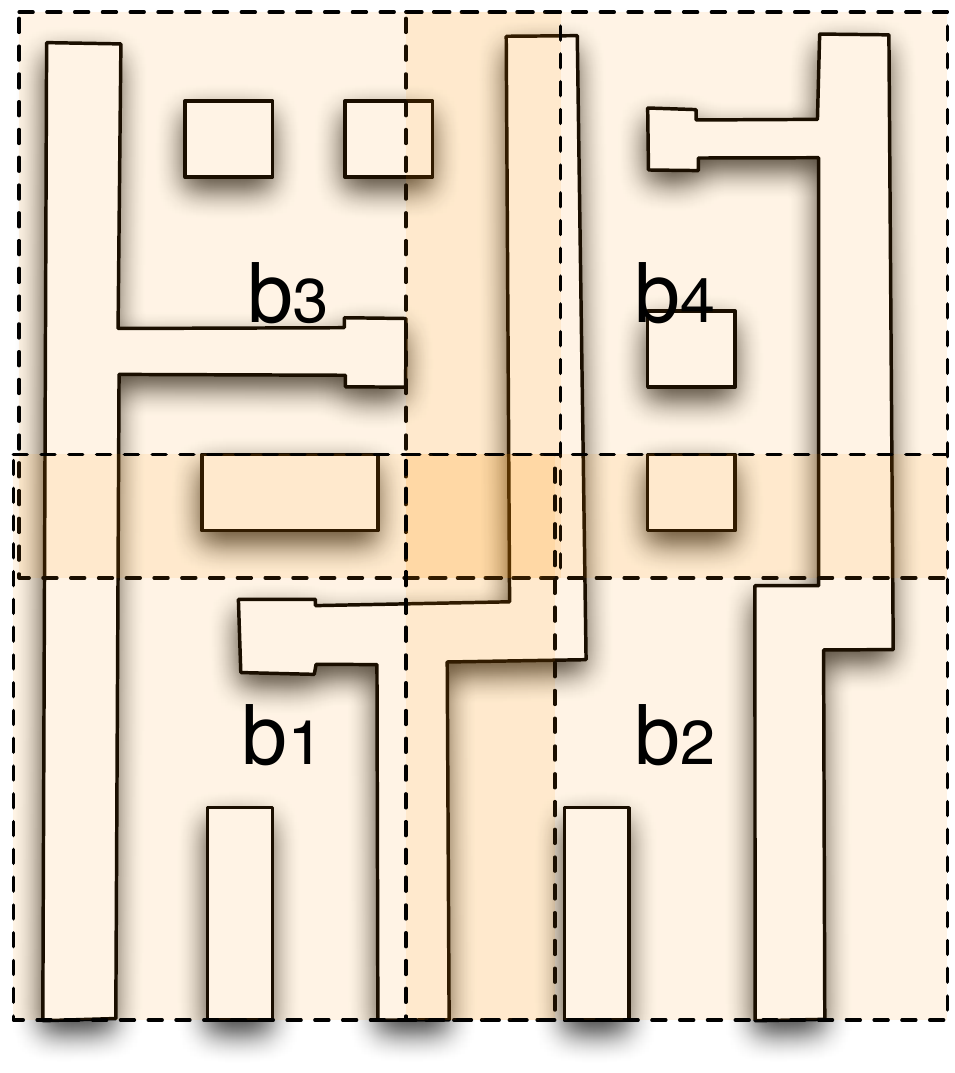}}
  \subfigure[] {\includegraphics[width=0.16\textwidth]{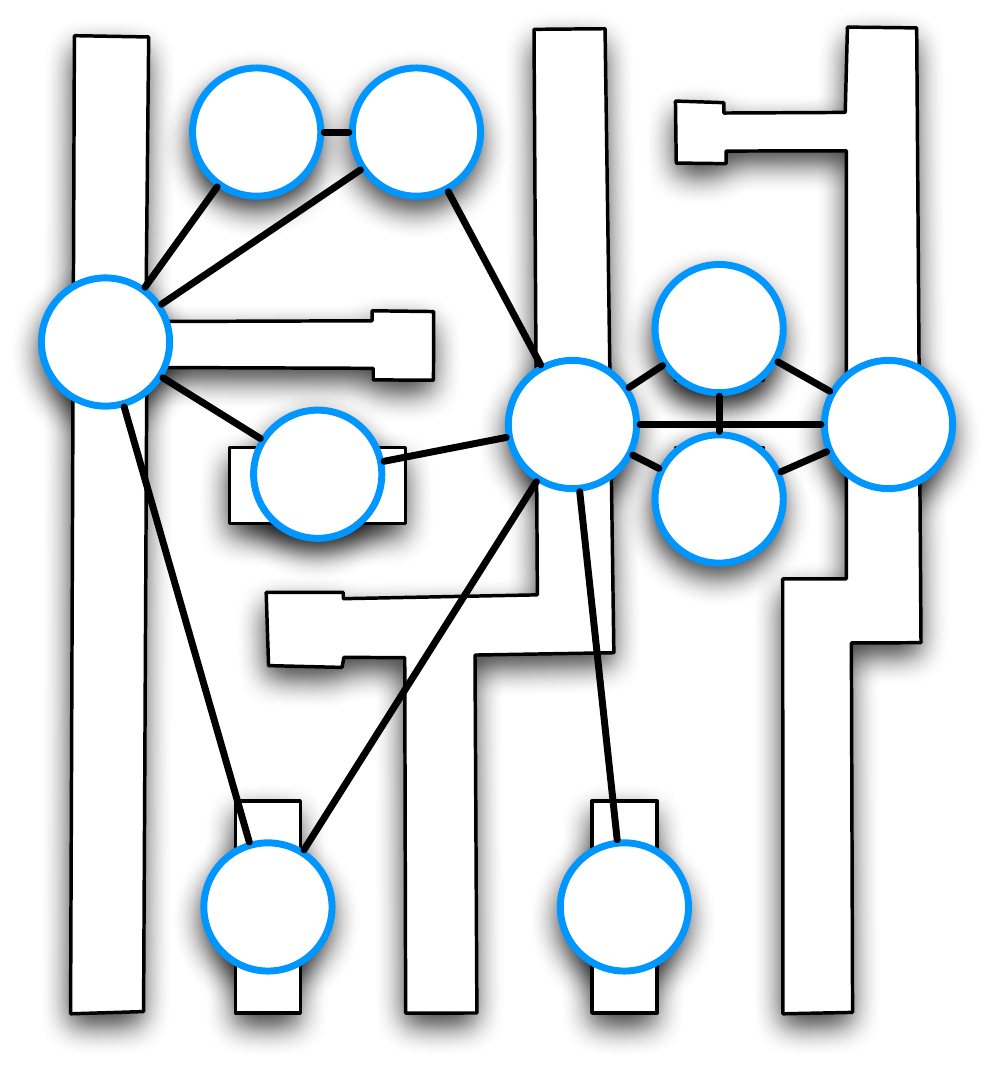}}
  \hspace{.1em}
  \subfigure[] {\includegraphics[width=0.16\textwidth]{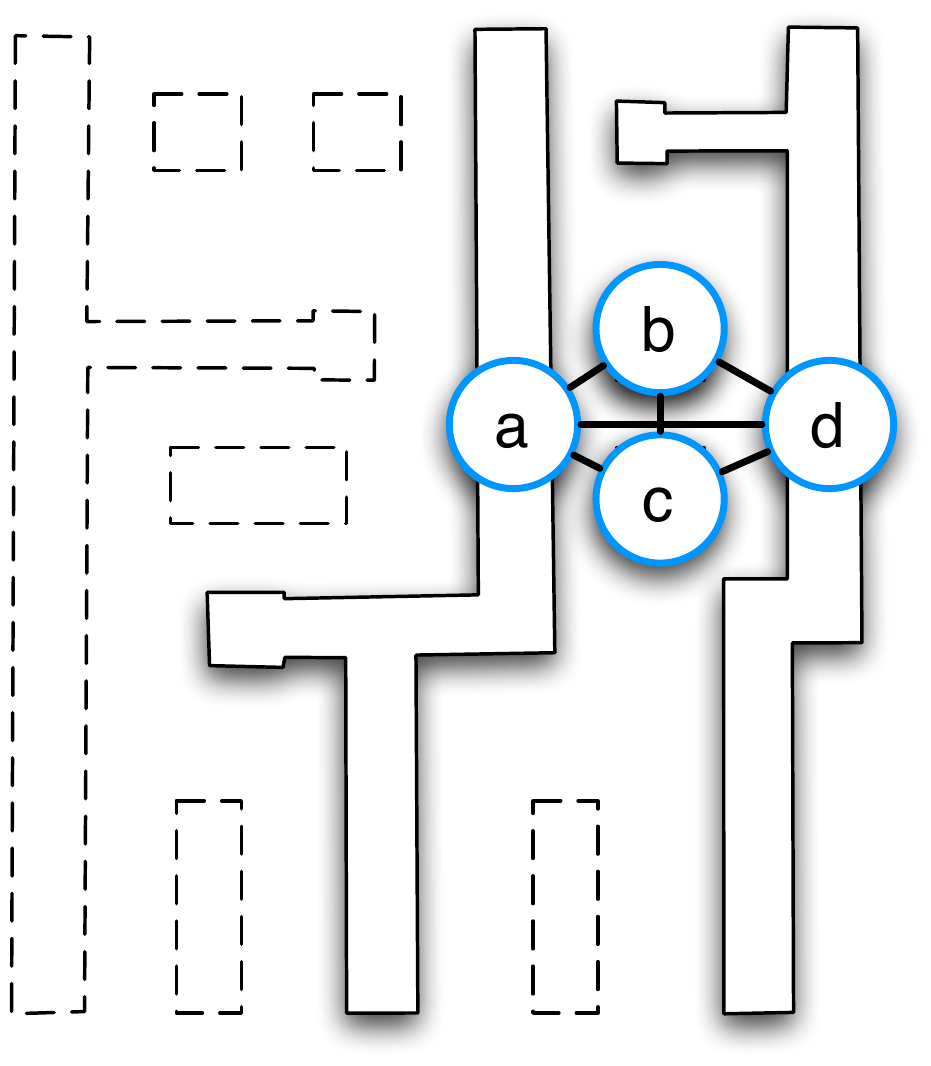}}
  %\subfigure[] {\includegraphics[width=0.14\textwidth]{DAC_ex3}}
  \subfigure[] {\includegraphics[width=0.15\textwidth]{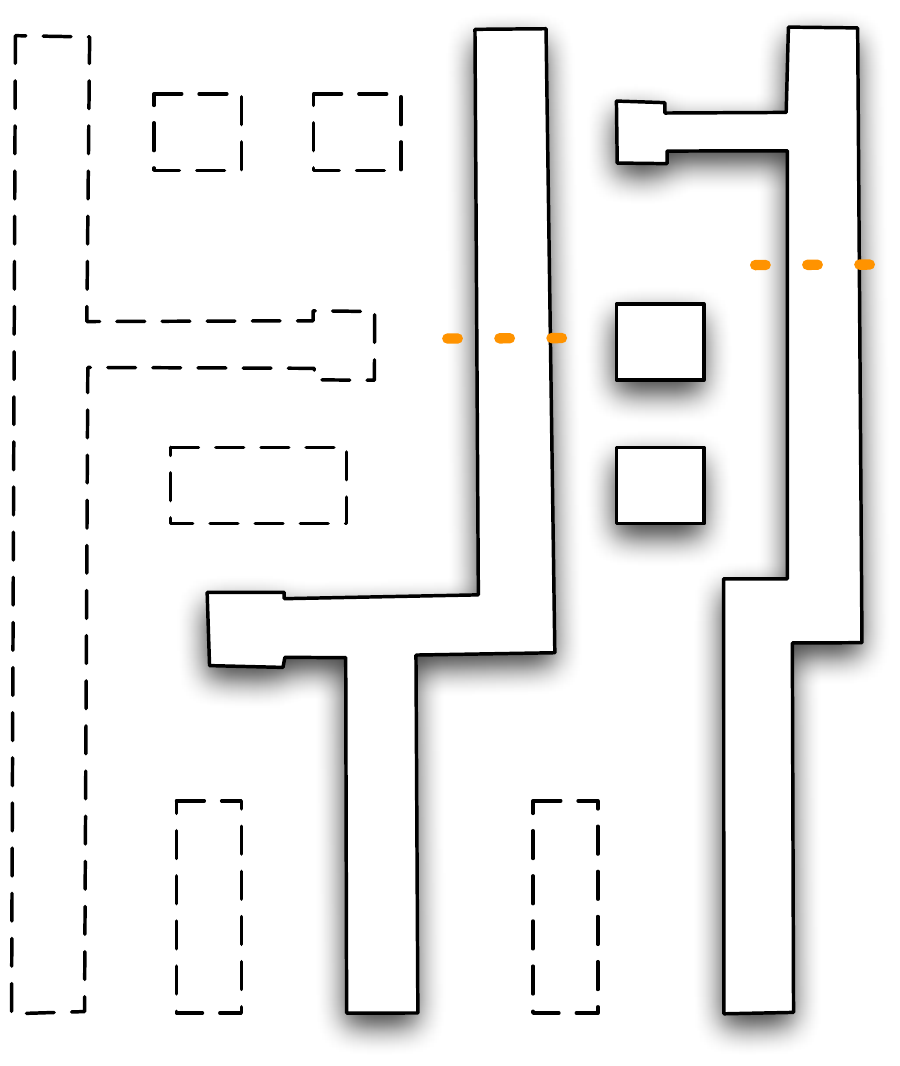}}
  \subfigure[] {\includegraphics[width=0.16\textwidth]{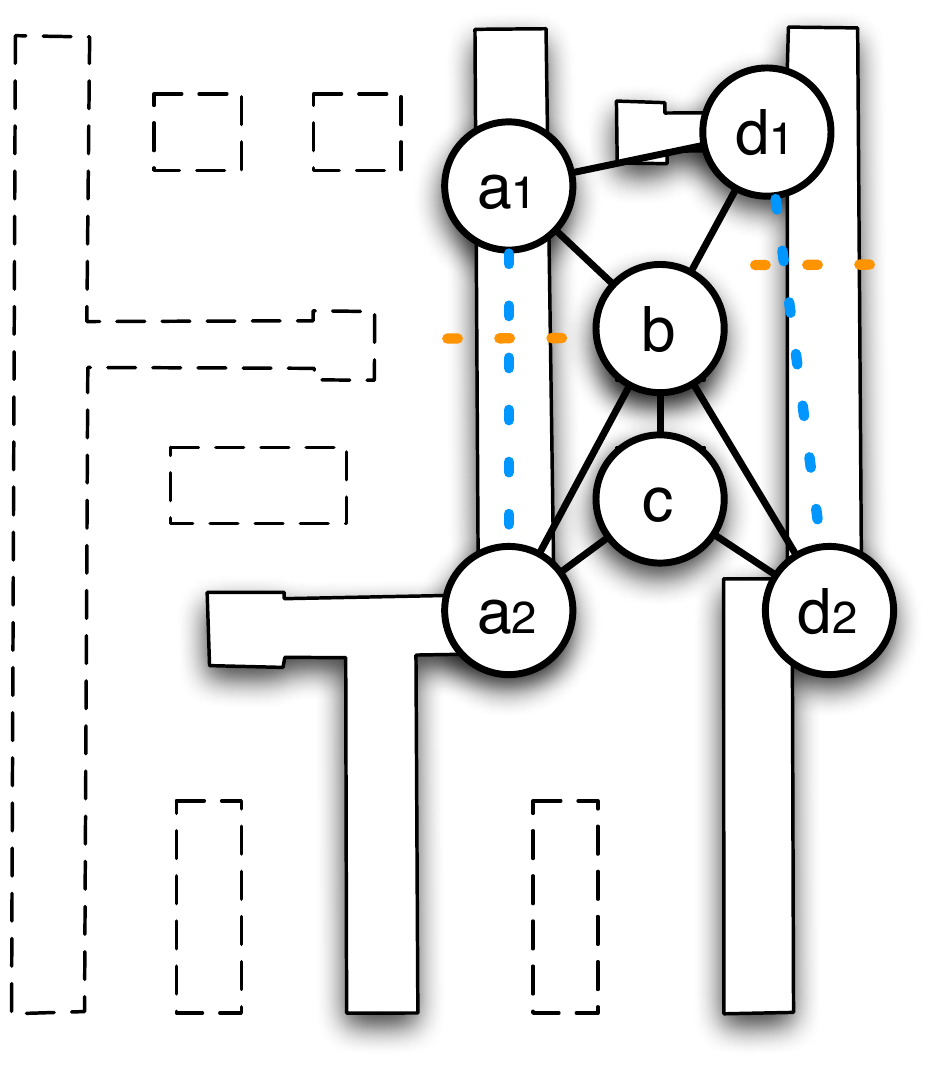}}
  \hspace{.1em}
  \subfigure[] {\includegraphics[width=0.16\textwidth]{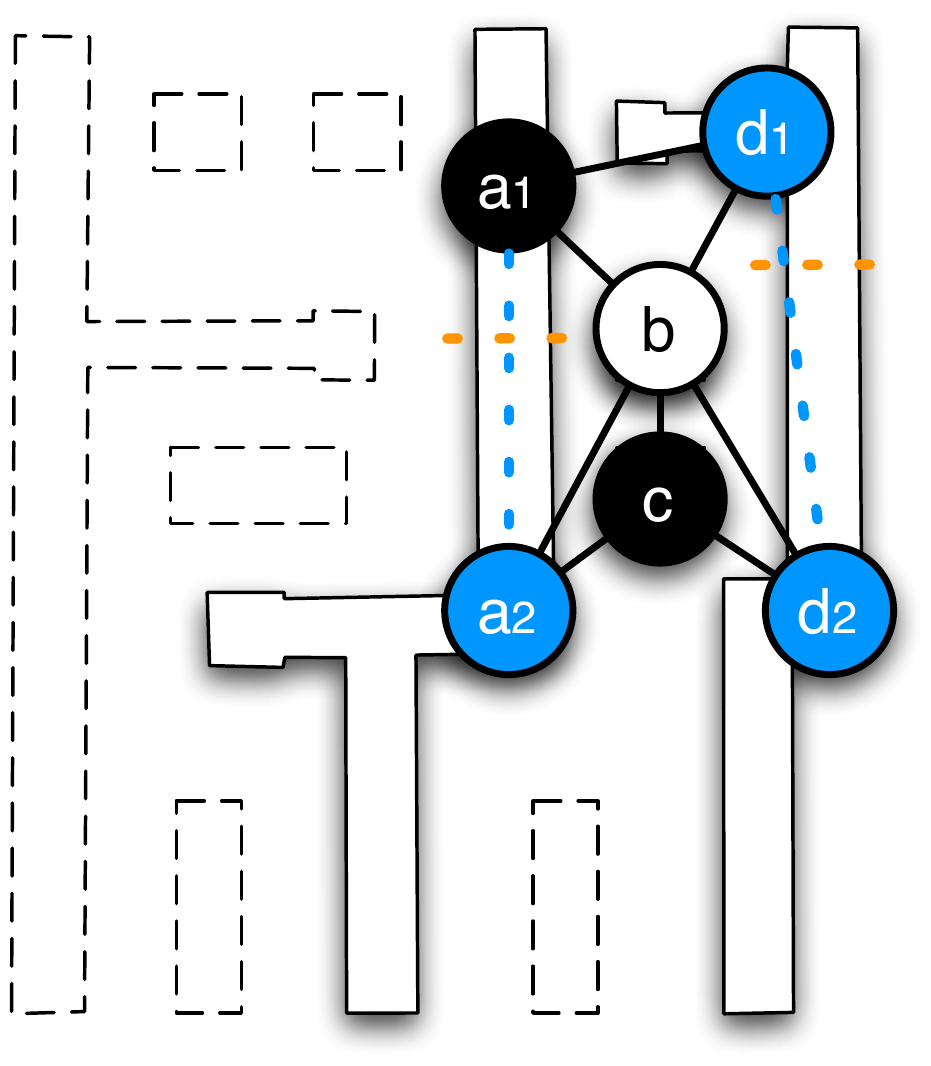}}
  \subfigure[] {\includegraphics[width=0.15\textwidth]{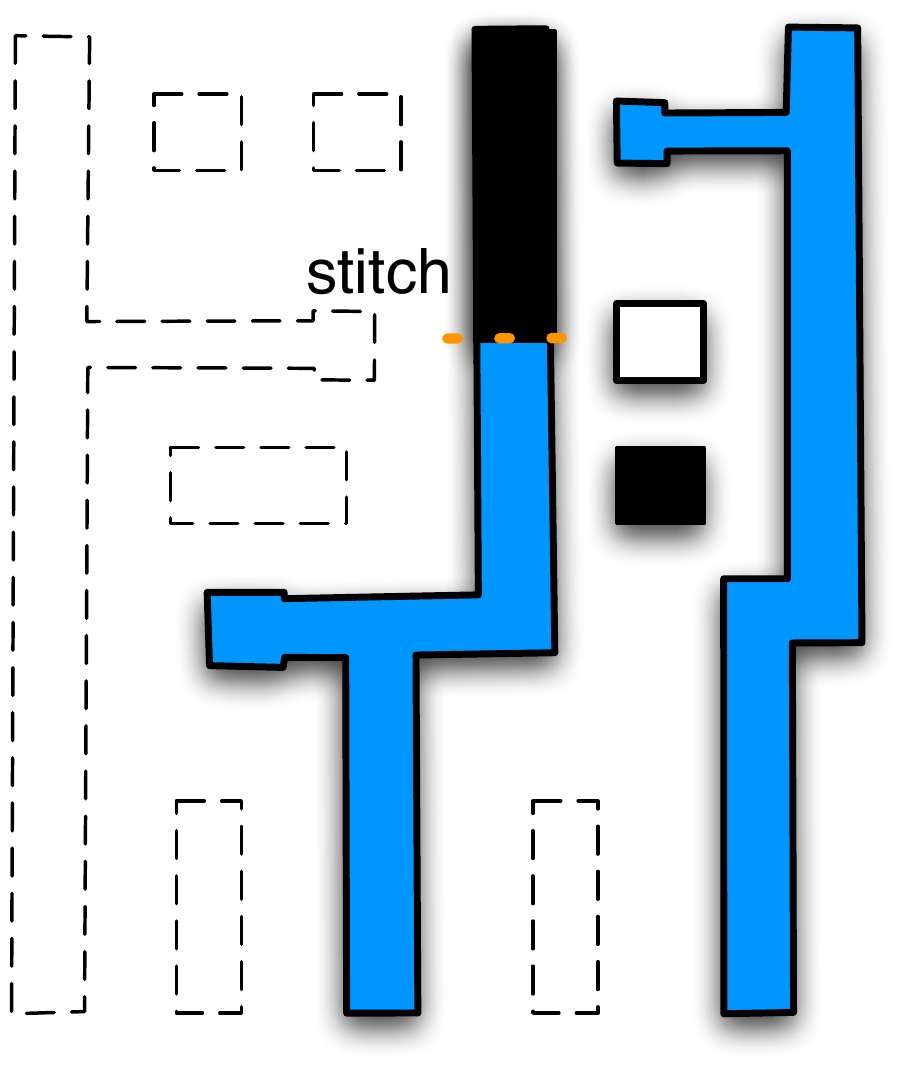}}
  \subfigure[] {\includegraphics[width=0.16\textwidth]{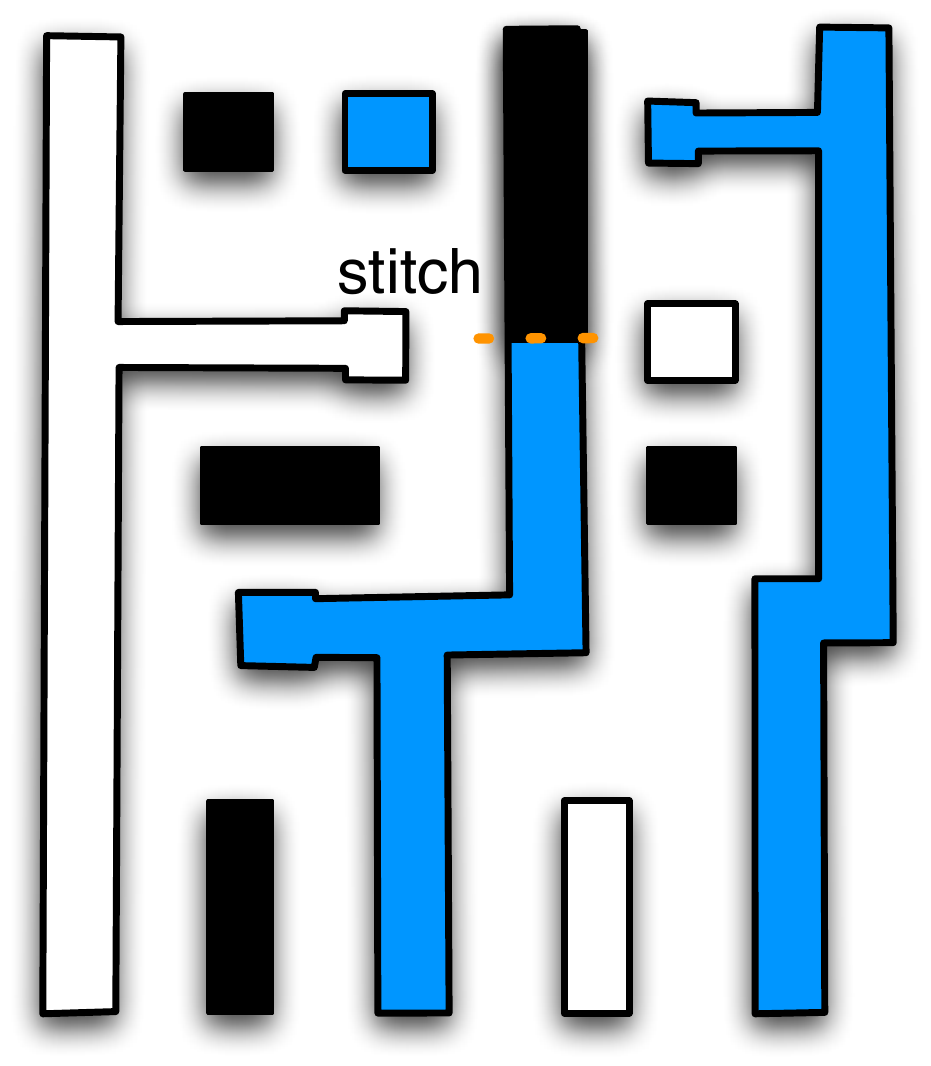}}
\iffalse
  \caption{~An example of the layout decomposition flow.~(a) Input layout.~(b) Layout graph construction.~(c) Simplified layout graph.
           (d) Projection.~(e) Stitch candidate generation.~(f) Decomposition graph construction.~(g) SDP formulation and partition based mapping on the decomposition graph.
           (h) Color assignment result for the decomposition graph.~(i) After iteratively recovering all the removed vertices, the decomposed layout.
  }
\fi
  \caption{~An example of the layout decomposition flow.}
  \label{fig:example}
  \vspace{-.4in}
\end{figure}

Fig. \ref{fig:example} illustrates an example to show the decomposition process step by step.
Given the input layout as in Fig. \ref{fig:example}(a), we partition it into a set of bins $\{b_1, b_2, b_3, b_4\}$ (see Fig. \ref{fig:example}(b)).
Then the layout graph is constructed (see Fig. \ref{fig:example}(c)), where the ten vertices representing the ten features in the input layout, and each vertex represents a polygonal feature (shape) where there is an edge (conflict edge) between two vertices if and only if those two vertices are within the minimum coloring distance $min_s$.
During the layout graph simplification, the vertices whose degree equal or smaller than two are iteratively removed from the graph.
The simplified layout graph, shown in Fig. \ref{fig:example}(d), only contains vertices $a, b, c$ and $d$.
Fig. \ref{fig:example}(d) shows the projection results.
Followed by stitch candidate generation \cite{TPL_DAC2013_Kuang}, there are two stitch candidates for TPL (see Fig. \ref{fig:example}(e)).
%Since the simplified layout graph is a four-clique graph, no matter how to assign mask, there would be one conflict.
%During our stitch candidate generation, because the projection sequences of vertices $a$ and $d$ include sub-sequences ``$212$", there are two stitch candidates for TPL.
Based on the two stitch candidates, vertices $a$ and $d$ are divided into two vertices, respectively.
The constructed decomposition graph is given in Fig. \ref{fig:example}(f).
It maintains all the information about conflict edges and stitch candidates, where the solid edges are the conflict edges while the dashed edges are the stitch edges and function as stitch candidates.
In each decomposition graph, a color assignment, which contains semidefinite programming (SDP) formulation and partition based mapping, is carried out.
During color assignment, the six vertices in the decomposition graph are assigned into three groups:
$\{a_1, c\}, \{b\}$ and $\{a_2, d_1, d_2\}$ (see Fig. \ref{fig:example}(g) and Fig. \ref{fig:example}(h)).
Here one stitch on feature $a$ is introduced.
After iteratively recover the removed vertices, the final decomposed layout is shown in Fig. \ref{fig:example}(i).
Our last process should be decomposition graphs merging, which combines the results on all decomposition graphs.
Since this example has only one decomposition graph, this process is skipped.

\vspace{-.1in}
\section{Density Balanced Decomposition}
\label{sec:algo}

%So far we have presented several core algorithms in our TPL layout decomposer.
%Since density balance is critical for CD and overlay control \cite{TPL_SPIE2012_Lucas},
%Different from all previous works, the density balance, especially local density balance, is considered throughout the flow.
Density balance, especially local density balance, is seamlessly integrated into each step of our decomposition flow.
In this section, we first elaborate how to integrate the density balance into the mathematical formulation and corresponding SDP formulation.
Followed by some discussion for density balance in all other steps.
%Note that although in this paper we discuss how to achieve balanced density on the whole layout,
%our algorithms can be easily extended to achieve good density balance in some specific areas.

% **********************************************************
%              Density Balanced SDP
% **********************************************************
\vspace{-.1in}
\subsection{Density Balanced SDP Algorithm}
\label{sec:db_sdp}

% Table of Notations
%{{{
\begin{table}[tb]
%\vspace{-.1in}
\renewcommand{\arraystretch}{1}
\centering
\caption{Notations used in color assignment}
\label{table:notation}
\begin{tabular}{|c|c|}
    \hline \hline
    %\multicolumn{2}{|c|}{Notations used in Mathematical programming}\\
    %\hline
    $CE$          & the set of conflict edges\ \ \ \ \ \\
    \hline
    $SE$          & the set of stitch edges\\
    \hline
    $V$           & the set of features\\
    \hline
    $B$           & the set of local bins\\
    %\hline
    %$r_i$         & the $i_{th}$ layout feature\\
    %\hline
    %$den_{ki}$    & the density of $r_i$ in bin $b_k$\\
    \hline\hline
\end{tabular}
\vspace{-.2in}
\end{table}
%}}}

For each decomposition graph, density balanced color assignment is carried out.
Some notations used are listed in Table \ref{table:notation}.
%If density balanced is not considered, \cite{TPL_ICCAD2011_Yu} presented a semidefinite programming (SDP) based algorithm.
See Appendix for some preliminary of semidefinite programming (SDP) based algorithm.

\vspace{-.05in}
\subsubsection{Density Balanced Mathematical Formulation}

%For each decomposition graph, we denote $CE$ as the set of conflict edges, $SE$ as the set of stitch edges, and $B$ as the set of local bins.
The mathematical formulation for the general density balanced layout decomposition is shown in (\ref{eq:math}),
where the objective is to simultaneously minimize the conflict number, the stitch number and the density uniformity of all bins.
Here $\alpha$ and $\beta$ are user-defined parameters for assigning the relative weights among the three values.

\begin{figure}[h]
\vspace{-.2in}
\begin{align}
  \label{eq:math}
  \textrm{min}      &   \sum_{e_{ij} \in CE} c_{ij} + \alpha \sum_{e_{ij} \in SE}s_{ij} + \beta \cdot \sum_{b_k \in B} DU_k \\
  \textrm{s.t}.\ \
    & c_{ij} = ( x_i == x_j )                       \qquad \qquad \qquad \      \forall e_{ij} \in CE    \tag{$1a$}\label{math_a}\\
    & s_{ij} = x_i \oplus x_j                       \qquad \qquad \qquad \qquad \forall e_{ij} \in SE    \tag{$1b$}\label{math_b}\\
    & x_i \in \{1, 2, 3\}                           \qquad \qquad \qquad \qquad \forall r_i \in V        \tag{$1c$}\label{math_c}\\
    & d_{kc} = \sum_{x_i = c} den_{ki}             \quad \qquad \qquad \forall r_i \in V, \ \ b_k \in B       \tag{$1d$}\label{math_d}\\
    & DU_k = \textrm{max}\{d_{kc}\} / \textrm{min}\{d_{kc}\}     \qquad  \forall b_k \in B               \tag{$1e$}\label{math_e}
\end{align}
\vspace{-.3in}
\end{figure}

Here $x_i$ is a variable representing the color (mask) of feature $r_i$, $c_{ij}$ is a binary variable for the conflict edge $e_{ij} \in CE$, and $s_{ij}$ is a binary variable for  the stitch edge $e_{ij} \in SE$.
The constraints (\ref{math_a}) and (\ref{math_b}) are used to evaluate the conflict number and stitch number, respectively.
%The constraint (\ref{math_a}) is used to evaluate the conflict number when the touched vertices $r_i$ and $r_j$ are assigned different masks.
%The constraint (\ref{math_b}) is used to calculate the stitch number. If vertices $r_i$ and $r_j$ are assigned the same mask, the stitch $s_{ij}$ is introduced.
The constraint (\ref{math_e}) is nonlinear, which makes the program (\ref{eq:math}) hard to be formulated into integer linear programming (ILP) as in \cite{TPL_ICCAD2011_Yu}.
Similar nonlinear constraints occur in the floorplanning problem \cite{FLOOR_DAC00_Chen}, where Tayor expansion is used to linearize the constraint into ILP.
However, Tayor expansion will introduce the penalty of accuracy.
Compared with the traditional time consuming ILP, semidefinite programming (SDP) has been shown to be a better approach in terms of runtime and solution quality tradeoffs \cite{TPL_ICCAD2011_Yu}.
However, how to integrate the density balance into the SDP formulation is still an open question.
In the following we will show that instead of using the painful Tayor expansion, this nonlinear constraint can be integrated into SDP without losing any accuracy.

\subsubsection{Density Balanced SDP Formulation}
%/*{{{*/

%\cite{TPL_ICCAD2011_Yu} shows that without the nonlinear constraint (\ref{math_e}), the program (\ref{eq:math}) can be formulated into a vector programming,
In SDP formulation,  the objective function is the representation of vector inner products, i.e., $\vec{v_i} \cdot \vec{v_j}$.
At the first glance, the constraint (\ref{math_e}) cannot be formulated into an inner product format.
However, we will show that density uniformity $DU_k$ can be optimized through considering another form
$DU_k^* = d_{k1} \cdot d_{k2} + d_{k1} \cdot d_{k3} + d_{k2} \cdot d_{k3}$.
This is based on the following observation: maximizing $DU_k^*$ is equivalent to minimizing $DU_k$.

%****  Proof of the observation
\iffalse
\begin{proof}
%{{{
First we prove that when $DU_k^*$ is maximum, $DU_k$ is minimum.
Since $d_{k1} + d_{k2} + d_{k3}$ is a constant $n$, after replacing $d_{k3}$ with $n - d_{k1} - d_{k2}$, we can achieve:
\begin{eqnarray}
DU_k^* & = & d_{k1} d_{k2} + d_{k1} (n - d_{k1} - d_{k2}) + d_{k2} (n - d_{k1} - d_{k2})  \notag\\
             & = & n (d_{k1} + d_{k2}) - d_{k1}^2 - d_{k2}^2 - d_{k1} d_{k2}              \notag
\end{eqnarray}
%If $DU_2$ gets the maximal value, then $\frac{\partial DU_2}{\partial d_1} = 0$ and $\frac{\partial DU_2}{\partial d_2} = 0$.
If $DU_k^*$ gets the maximal value, then $\partial DU_k^* / \partial d_{k1} = 0$ and $\partial DU_k^* / \partial d_{k2} = 0$.
\begin{displaymath}
	\left\{
    \begin{array}{c}
    n - 2 \cdot d_{k1} - d_{k2} = 0 \\
    n - 2 \cdot d_{k2} - d_{k1} = 0
    \end{array}
	\right.
  \Rightarrow
  d_{k1} = d_{k2} = d_{k3}
\end{displaymath}
Besides, $\partial^2 DU_k^* / \partial d_{k1}^2 = -2$, which means $DU_k^*$ is maximum.
Meanwhile, $DU_k = \textrm{max} \{d_{kc}\} / \textrm{min} \{d_{kc}\}$.
It is easy to see that $DU_k$ achieve minimal value $1$.

From the other direction, if $DU_k$ is minimum, i.e. $1$, then $d_{k1} = d_{k2} = d_{k3}$.
As discussed above , $DU_k^*$ is maximum.
%}}}
\end{proof}
\fi

%Since maximizing $DU_k^*$ is equivalent to minimizing $DU_k$, if $DU_k^*$ can represented as vector inner products,
%we find a way to integrate the constraint (\ref{math_e}) into vector program.

\begin{mylemma}
$DU_k^* = 2/3 \cdot \sum_{i, j \in V} den_{ki} \cdot den_{kj} \cdot (1 - \vec{v_i} \cdot \vec{v_j})$,
where $den_{ki}$ is the density of feature $r_i$ in bin $b_k$.
\label{lem:2}
\vspace{-.1in}
\end{mylemma}
%See Appendix \ref{sec:app_proof} for proof of Lemma \ref{lem:2}.

% *********** Proof of Lemma
\begin{proof}
%{{{
First of all, let us calculate $d_1 \cdot d_2$.
For all vectors $\vec{v_i} = (1, 0)$ and all vectors $\vec{v_j} = (-\frac{1}{2}, \frac{\sqrt{3}}{2})$,
we can see that
\begin{align}
  & \sum_i \sum_j len_i \cdot len_j \cdot (1 - \vec{v_i} \cdot \vec{v_j}) = \sum_i \sum_j len_i \cdot len_j \cdot 3/2 \notag\\
= & 3/2 \cdot \sum_i len_i \sum_j len_j = 3/2 \cdot d_1 \cdot d_2 \notag
\end{align}
So $d_1 \cdot d_2 = 2/3 \cdot \sum_i \sum_j len_i \cdot len_j \cdot (1 - \vec{v_i} \cdot \vec{v_j})$, where $\vec{v_i} = (1, 0)$ and $\vec{v_j} = (-\frac{1}{2}, \frac{\sqrt{3}}{2})$.
We can also calculate $d_1 \cdot d_3$ and $d_2 \cdot d_3$ using similar methods.
Therefore, 
\begin{align}
  DU_2 & =  d_1 \cdot d_2 + d_1 \cdot d_3 + d_2 \cdot d_3                                       \notag\\
       & =  2/3 \cdot \sum_{i, j \in V} len_i \cdot len_j \cdot (1 - \vec{v_i} \cdot \vec{v_j}) \notag
\end{align}
%\vspace{-.1in}
%}}}
\end{proof}

Because of Lemma \ref{lem:2}, the $DU_k^*$ can be represented as a vector inner product,
then we have achieved the following theorem.
%Combining Lemma \ref{lem:1} and Lemma \ref{lem:2}, we have proved the following theorem.

\begin{mytheorem}
Maximizing $DU_k^*$ can achieve better density balance in bin $b_k$.
\vspace{-.1in}
\end{mytheorem}

Note that we can remove the constant $\sum_{i,j \in V} den_{ki} \cdot den_{kj} \cdot 1$ in $DU_k^*$ expression.
Similarly, we can eliminate the constants in the calculation of the conflict and stitch numbers.
The simplified vector program is as follows:

\begin{figure}[h]
\vspace{-.1in}
\begin{align}
  \vspace{-.1in}
  \label{eq:vp}
  \textrm{min}  & \sum_{e_{ij} \in CE} ( \vec{v_i} \cdot \vec{v_j} ) - \alpha \sum_{e_{ij} \in SE} ( \vec{v_i} \cdot \vec{v_j} ) - \beta \cdot \sum_{b_k \in B} DU_k^*\\
  \textrm{s.t}.\ \
    & DU_k^* = - \sum_{i, j \in V} den_{ki} \cdot den_{kj} \cdot ( \vec{v_i} \cdot \vec{v_j})     \quad \forall b_k \in B             \label{vp_a}\tag{$2a$}\\
    & \vec{v_i} \in \{(1, 0), (-\frac{1}{2}, \frac{\sqrt{3}}{2}), (-\frac{1}{2}, -\frac{\sqrt{3}}{2})\}     \label{vp_b}\tag{$2b$}
\end{align}
\vspace{-.2in}
\end{figure}

Formulation (\ref{eq:vp}) is equivalent to the mathematical formulation (\ref{eq:math}), and it is still NP-hard to be solved exactly.
Constraint (\ref{vp_b}) requires the solutions to be discrete.
To achieve a good tradeoff between runtime and accuracy, we can relax (\ref{eq:vp}) into a SDP formulation, as shown in Theorem \ref{the:sdp}.

\begin{mytheorem}
\vspace{0in}
\label{the:sdp}
Relaxing vector program (\ref{eq:vp}) can get the SDP formulation (\ref{eq:sdp}).
\vspace{-.1in}
\end{mytheorem}

\begin{figure}[h]
\vspace{-.1in}
\begin{align}
    \label{eq:sdp}
    \textrm{SDP:\ \ min}    & \ \ \ A \bullet X                     \\
                & X_{ii} = 1,              \ \ \forall i \in V         \tag{$3a$}\\
                & X_{ij} \ge -\frac{1}{2}, \ \ \forall e_{ij} \in CE   \tag{$3b$}\\
                & X \succeq 0                                          \tag{$3c$}\label{sdp_c}
\end{align}
\vspace{-.3in}
\end{figure}
where $A_{ij}$ is the entry that lies in the $i$-th row and the $j$-th column of matrix $A$:
\begin{equation}
    A_{ij} =
    \left\{
    \begin{array}{cc}
        1 + \beta \cdot \sum_{k} den_{ki} \cdot den_{kj},                  & \forall b_k \in B, e_{ij} \in CE\\
        -\alpha + \beta \cdot \sum_{k} den_{ki} \cdot den_{kj},         & \forall b_k \in B, e_{ij} \in SE\\
        \beta \cdot \sum_{k} den_{ki} \cdot den_{kj},                        & \textrm{otherwise}
    \end{array}
    \right.
    \notag
\label{eq:sdp_a}
\end{equation}
%Here $A_{ij}$ is the entry that lies in the $i$-th row and the $j$-th column of matrix $A$.

% *********  Proof of Theorem
\iffalse
\begin{proof}
%{{{
To achieve a good tradeoff between runtime and accuracy, we relax constraint (\ref{vp_b}) to generate formula (\ref{eq:rvp}) as following:
\begin{align}
  \vspace{-.1in}
  \textrm{min}  & \sum_{e_{ij} \in CE} ( \vec{y_i} \cdot \vec{y_j}) - \alpha \sum_{e_{ij} \in SE} ( \vec{y_i} \cdot \vec{y_j} ) - \beta \cdot \sum_{b_k \in B} DU_k^* \label{eq:rvp}\\
  \textrm{s.t}.\ \
    & DU_k^* = - \sum_{i, j \in V} den_{ki} \cdot den_{kj} \cdot (\vec{y_i} \cdot \vec{y_j})   \ \ \forall b_k \in B   \label{sdpa}\tag{$a$}\\
    & \vec{y_i} \cdot \vec{y_i} = 1 ,               \ \ \ \forall i \in V                   \label{sdpb}\tag{$b$}\\
    & \vec{y_i} \cdot \vec{y_j} \ge -\frac{1}{2},   \ \ \ \forall e_{ij} \in CE             \label{sdpc}\tag{$c$}
\end{align}
Here each unit vector $\vec{v_i}$ is replaced by a $n$ dimensional vector $\vec{y_i}$.
In the following we prove that SDP formulation (\ref{eq:sdp}) and relaxed vector programming (\ref{eq:rvp}) are equivalent.

Given solutions $\{\vec{v_1}, \vec{v_2}, \cdots \vec{v_m}\}$ of (\ref{eq:rvp}), the corresponding matrix $X$ is defined as $X_{ij}=\vec{v_i} \cdot \vec{v_j}$.
In the other direction, given a matrix $X$ from (\ref{eq:sdp}),  we can find a matrix $V$ satisfying $X=VV^T$ by using the Cholesky decomposition.
The rows of $V$ are vectors $\{v_i\}$ that form the solutions of (\ref{eq:rvp}).
%}}}
\end{proof}
\fi

Due to space limit, the detailed proof is omitted.
%The proof of Theorem \ref{the:sdp} is provided in Appendix \ref{sec:app_proof}.
The solution of (\ref{eq:sdp}) is continuous instead of discrete, and provides a lower bound of vector program (\ref{eq:vp}).
In other words, (\ref{eq:sdp}) provides an approximated solution to (\ref{eq:vp}).
\subsection{Density Balanced Mapping}
\label{sec:db_mapping}

Each $X_{ij}$ in solution of (\ref{eq:sdp}) corresponds to a feature pair $(r_i, r_j)$.
The value of $X_{ij}$  provides a guideline, i.e., whether two features $r_i$ and $r_j$ should be in same color.
If $X_{ij}$ is close to $1$, features $r_i$ and $r_j$ tend to be in the same color (mask);
while if it is close to $-0.5$, $r_i$ and $r_j$ tend to be in different colors (masks).
With these guidelines a mapping procedure is adopted to finally assign all input features into three colors (masks).

%\vspace{-.05in}
\subsubsection{Limitations of Greedy Mapping}

In \cite{TPL_ICCAD2011_Yu}, a greedy approach was applied for the final color assignment.
The idea is straightforward: all $X_{ij}$ values are sorted, and vertices $r_i$ and $r_j$ with larger $X_{ij}$ value tend to be in the same color.
The $X_{ij}$ can be classified into two types: clear and vague.
If most of the $X_{ij}$s in matrix $X$ are clear (close to 1 or -0.5), this greedy method may achieve good result.
However, if the decomposition graph is not 3-colorable, some values in matrix $X$ are vague.
For the vague $X_{ij}$, e.g., 0.5, the greedy method may not be so effective.

%\vspace{-.1in}
\subsubsection{Density Balanced Partition based Mapping}

Contrary to the previous greedy approach, we propose a partition based mapping, which can solve the assignment problem for the vague $X_{ij}$s in a more effective way.
The new mapping is based on a three-way maximum-cut partitioning.
The main ideas are as follows.
%If a $X_{ij}$ is clear, the relationship of vertices $r_i$ and $r_j$ can be directly decided.
If a $X_{ij}$ is vague, instead of only relying on the SDP solution, we also take advantage of the information in decomposition graph.
The information is captured through constructing a graph, denoted by $G_M$.
Through formulating the mapping as a three-way partitioning on the graph $G_M$, our mapping can provide a global view to search better solutions.

\begin{algorithm}[thb]
\caption{Partition based Mapping}
\label{alg:mapping}
\begin{algorithmic}[1]
   \REQUIRE Solution matrix $X$ of the program (\ref{eq:sdp}).
   \STATE Label each non-zero entry $X_{i, j}$ as a triplet $(X_{ij}, i, j)$;
   \STATE Sort all $(X_{ij}, i, j)$ by $X_{ij}$;
   \FOR { all triples with $X_{ij} > th_{unn}$}
      \STATE Union(i, j);
   \ENDFOR
   \FOR { all triples with $X_{ij} < th_{sp}$}
      \STATE Separate(i, j);
   \ENDFOR
   \STATE Construct graph $G_M$;
   \IF{graph size $\le$ 3}
      \STATE return;
   \ELSIF{graph size $\le 7$}
      \STATE Backtracking based three-way partitioning;
   \ELSE
      \STATE FM based three-way partitioning;
   \ENDIF
\end{algorithmic}
\end{algorithm}

% Figure of density balanced mapping
\begin{figure}[tb]
    \centering
    %\hspace{-.2in}
    \subfigure[]{\includegraphics[width=0.16\textwidth]{DAC13_DG1}}
    \hspace{.1in}
    \subfigure[]{\includegraphics[width=0.16\textwidth]{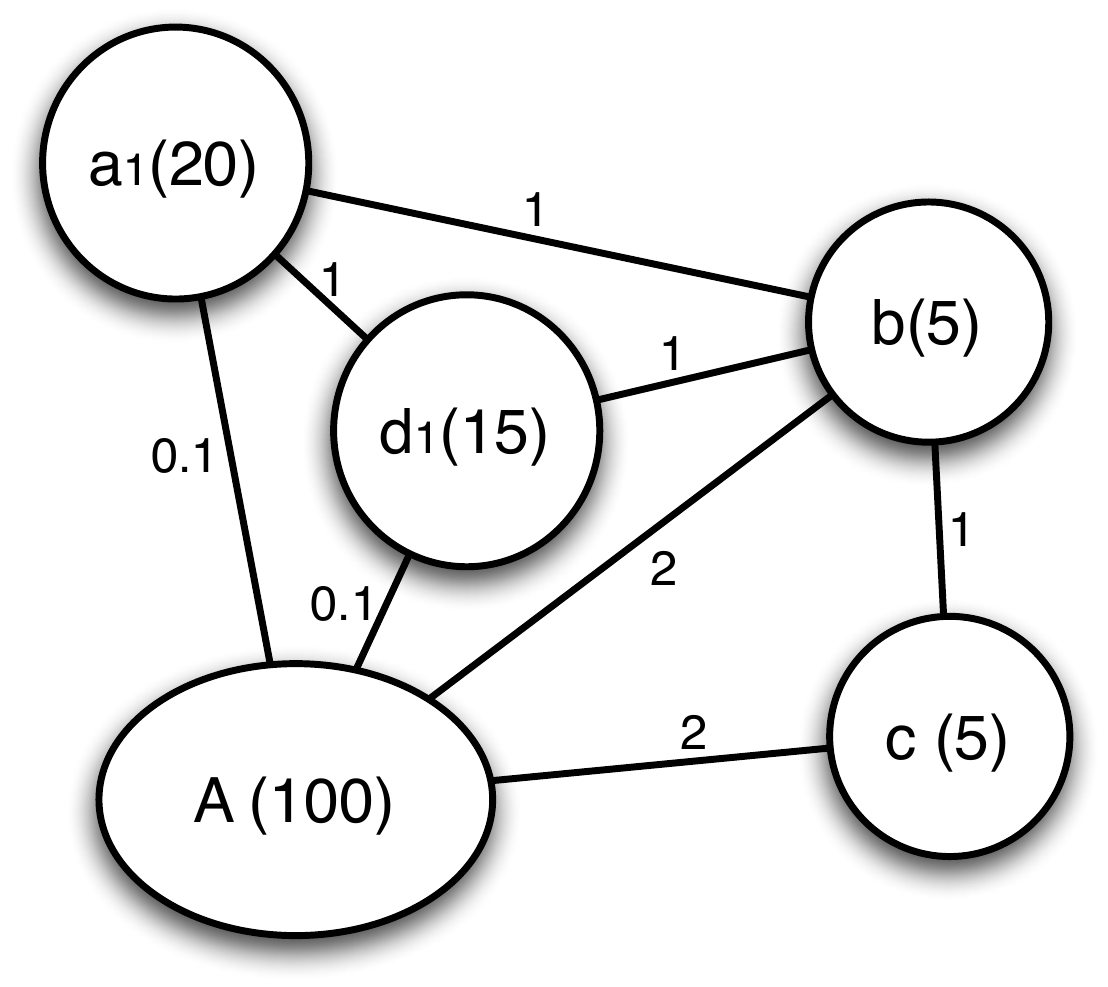}}
    %\hspace{2in}
    \subfigure[]{\includegraphics[width=0.16\textwidth]{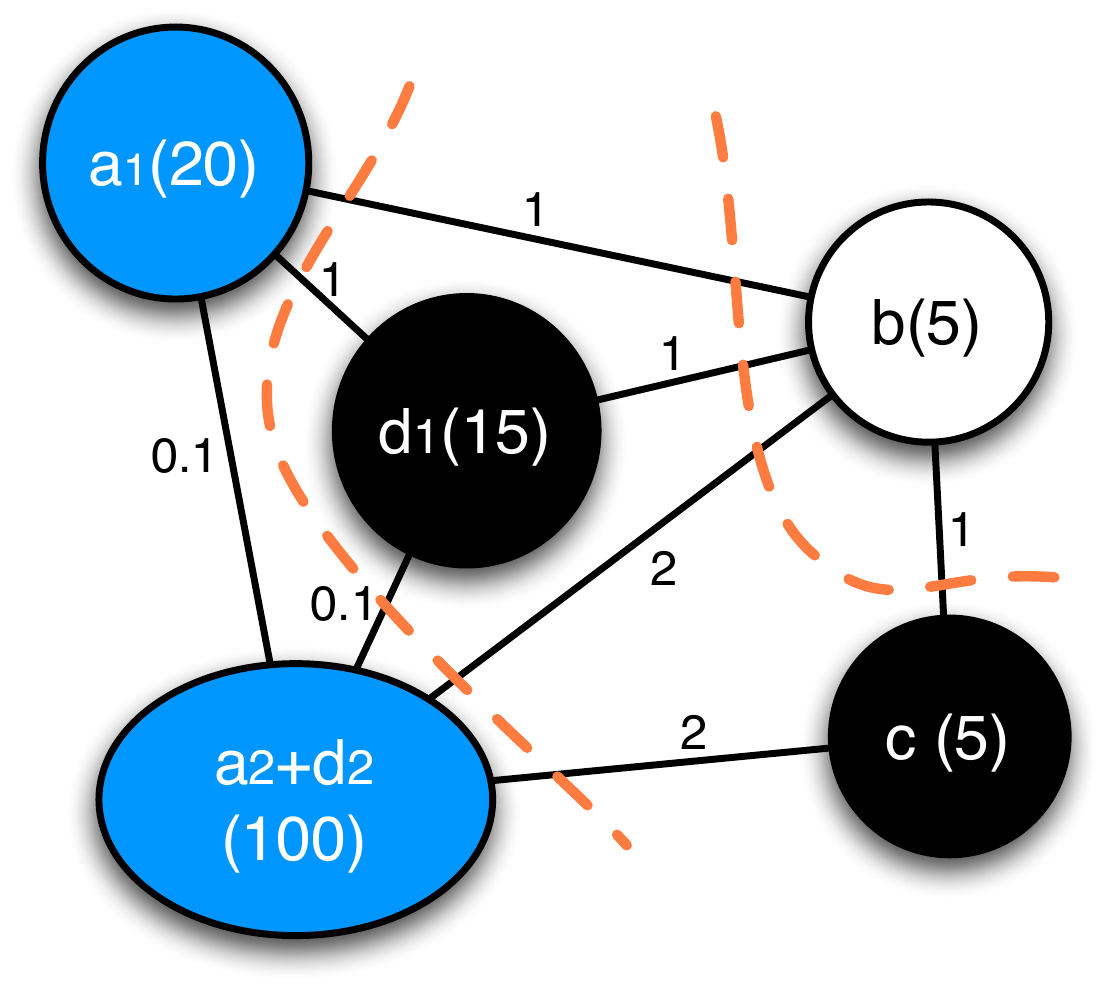}}
    %\hspace{-.1in}
    \subfigure[]{\includegraphics[width=0.16\textwidth]{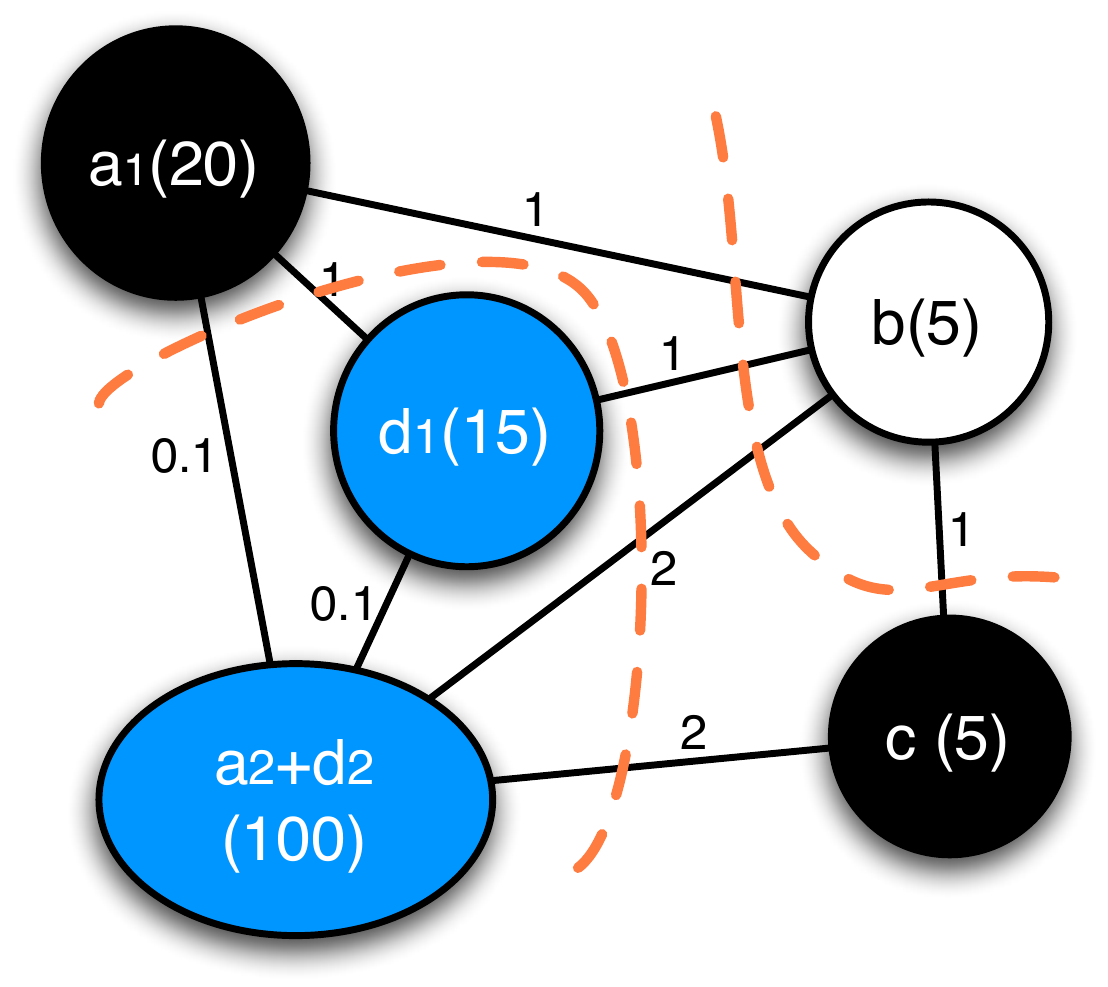}}
    %\hspace{-.2in}
    \vspace{-.1in}
    \caption{
    Density Balanced Mapping.
    ~(a) Decomposition graph.~(b) Construct graph $G_M$.
    ~(c) Mapping result with cut value 8.1 and density uniformities 24.
    ~(d) A better mapping with cut 8.1 and density uniformities 23.
    }
    \label{fig:db_mapping}
    \vspace{-.1in}
\end{figure}

Algorithm \ref{alg:mapping} shows our partition based mapping procedure.
Given the solutions from program (\ref{eq:sdp}), some triplets are constructed and sorted to maintain all non-zero $X_{ij}$ values (lines 1--2).
The mapping incorporates two stages to deal with the two different types.
The first stage (lines 3--8) is similar to that in \cite{TPL_ICCAD2011_Yu}.
If $X_{ij}$ is clear then the relationship between vertices $r_i$ and $r_j$ can be directly determined.
Here $th_{unn}$ and $th_{sp}$ are user-defined threshold values.
%They determine what kind of $X_{ij}$ belongs to the clear type.
For example, if $X_{ij} > th_{unn}$, which means that $r_i$ and $r_j$ should be in the same color,
then function Union(i, j) is applied to merge them into a large vertex.
Similarly, if $X_{ij} < th_{sp}$, then function Separate(i, j) is used to label $r_i$ and $r_j$ as incompatible.
%If $r_i$ and $r_j$ are incompatible, they cannot be merged and function Compatible(i, j) will return $false$.
In the second stage (lines 9--16) we deal with the vague $X_{ij}$ values.
%<CLARIFICATION: can you give an estimate on "not large" in terms of starting size? as in, G is the set of all groups, and |G| ~= 10% |V| or something?>
During the previous stage some vertices have been merged, therefore the total vertex number is not large.
Here we construct a graph $G_M$ to represent the relationships among all the remanent vertices (line 9).
Each edge $e_{ij}$ in this graph has a weight representing the cost if vertices $i$ and $j$ are assigned into same color.
Therefore, the color assignment problem can be formulated as a maximum-cut partitioning problem on $G_M$ (line 10--16).
%Note that the graph can further consider density balance by adding a weight on each vertex to represent its density.

Through assigning a weight to each vertex representing its density, graph $G_M$ is able to balance density among different bins.
%Through assigning each vertex with a weight representing its density, the density balance can be naturally integrated in the graph $G_M$.
Based on the $G_M$, a partitioning is performed to simultaneously achieve a maximum-cut and balanced weight among different parts.
Note that we need to modify the gain function, then in each move, we try to achieve a more balanced and larger cut partitions.
%Therefore, density balance has been integrated into the mapping.

An example of the density balanced mapping is shown in Fig. \ref{fig:db_mapping}.
Based on the decomposition graph (see Fig. \ref{fig:db_mapping} (a)), SDP is formulated.
Given the solutions of SDP, after the first stage of mapping, vertices $a_2$ and $d_2$ are merged in to a large vertex.
As shown in Fig. \ref{fig:db_mapping}(b), the graph $G_M$ is constructed, where each vertex is associated with a weight.
There are two partition results with the same cut value 8.1 (see Fig. \ref{fig:db_mapping} (c) and Fig. \ref{fig:db_mapping} (d)).
However, their density uniformities are 24 and 23, respectively.
To keep a more balanced density result, the second partitioning in Fig. \ref{fig:db_mapping} (c) is adopted as color assignment result.

%\begin{define}[Group Graph]
%The group graph is an undirected graph whose vertices are the group vertices.
%An edge exists if and only if there are conflict edges or stitch edges between the two group vertices.
%\end{define}

It is well known that the maximum-cut problem, even for a 2-way partition, is NP-hard.
However, we observe that in many cases, after the global SDP optimization, the graph size of $G_M$ could be quite small, i.e., less than 7.
For these small cases, we develop a backtracking based method to search the entire solution space.
Note that here backtracking can quickly find the optimal solution even through three-way partitioning is NP-hard.
If the graph size is larger, we propose a heuristic method, motivated by the classic FM partitioning algorithm
\cite{PAR_DAC82_FM}\cite{PAR_TC89_Sanchis}.
Different from the classic FM algorithm, we make the following modifications.
(1) In the first stage of mapping, some vertices are labeled as incomparable, therefore before moving a vertex from one partition to another, we should check whether it is legal.
(2) Classical FM algorithm is for min-cut problem, we need to modify the gain function of each move to achieve a maximum cut.

The runtime complexity of graph construction is $O(m)$, where $m$ is the vertex number in $G_M$.
The runtime of three-way maximum-cut partitioning algorithm is $O(m logm)$.
Besides, the first stage of mapping needs $O(n^2logn)$ \cite{TPL_ICCAD2011_Yu}.
Since $m$ is much smaller than $n$, the complexity of density balanced mapping is $O(n^2logn)$.

\iffalse
%{{{
Here an important concept is called the ``vertex move".
%Each vertex has a gain value showing how much cut can be increased if it is moved from current partition to another one.
At every move, the vertex who can achieve the most cut increasing is
chosen to move from current partition to another one. Another
important concept is called the ``pass", where all vertices are
moved exactly once during a single pass. When a vertex is chosen to
move because it has the maximum cut increasing, we move it even when
the increasing value is negative. At the end of the current pass, we
accept the first $K$ movings that lead to the best partition, which
may contain some not so good movings.

% ************************************
%       Runtime analyse
% ************************************
The Disjoint-set data structure is used to implement functions
Union() and Separate(). Under a special implementation, the runtime of each function can be almost constant \cite{book90Algorithm}.
Let $n$ be the vertex number, then the number of triplets is $O(n^2)$.
Sorting all the triplets requires $O(n^2logn)$.
The runtime complexity of graph construction is $O(m)$, where $m$ is the vertex number in $G_M$.
The runtime of three-way maximum-cut partitioning algorithm is $O(m logm)$.
Since $m$ is much smaller than $n$, then the complexity of density balanced mapping is $O(n^2logn)$.
%}}}
\fi

% **********************************************************
%              Density Balanced Simplification
% **********************************************************
\vspace{-.1in}
\subsection{Density Balanced Layout Graph Simplification}

Here we show that the layout graph simplification, which was proposed in \cite{TPL_ICCAD2011_Yu}, can consider the local density balance as well.
During layout graph simplification, we iteratively remove and push all vertices with degree less than or equal to two.
After the color assignment on the remained vertices, we iteratively recover all the removed vertices and assign legal colors.
Instead of randomly picking one, we search a legal color which is good for the density uniformities.

\begin{comment}
The second technique is presented during decomposition graphs merging.
Using bridge computation \cite{TPL_DAC2012_Fang}, the whole design can be broken down into several components.
After color assignment independently on each component, the results on all decomposition graphs are merged together.
Since that two components share one or two vertex/edge(s), we can rotate the colors in one component so that the merging can introduces better balanced density.
Note that there is no cycle when we merge the decomposition graphs, thus the merging can be done in linear time through a breadth-first search (BFS).
\end{comment}

\vspace{-.1in}
\section{Speedup Techniques}
\label{sec:speedup}

Our layout decomposer applies a set of graph simplification techniques proposed by recent works:
\begin{itemize}
  \item Independent Component Computation \cite{TPL_ICCAD2011_Yu}\cite{TPL_DAC2012_Fang}\cite{TPL_DAC2013_Kuang}; 
  \item Vertex with Degree Less than 3 Removal \cite{TPL_ICCAD2011_Yu}\cite{TPL_DAC2012_Fang}\cite{TPL_DAC2013_Kuang};
  \item 2-Edge-Connected Component Computation \cite{TPL_ICCAD2011_Yu}\cite{TPL_DAC2012_Fang}\cite{TPL_DAC2013_Kuang};
  \item 2-Vertex-Connected Component Computation \cite{TPL_DAC2012_Fang}\cite{TPL_DAC2013_Kuang}.
\end{itemize} 
Apart from the above graph simplifications, our decomposer proposes a set of novel speedup techniques, which would be introduced in this section.

\vspace{-.1in}
\subsection{LG Cut Vertex Stitch Forbiddance}
\label{sec:nostitch}

\begin{figure}[hbt]
    \centering
    \vspace{-.1in}
    %\hspace{-.6in}
    \subfigure[]{\includegraphics[width=0.16\textwidth]{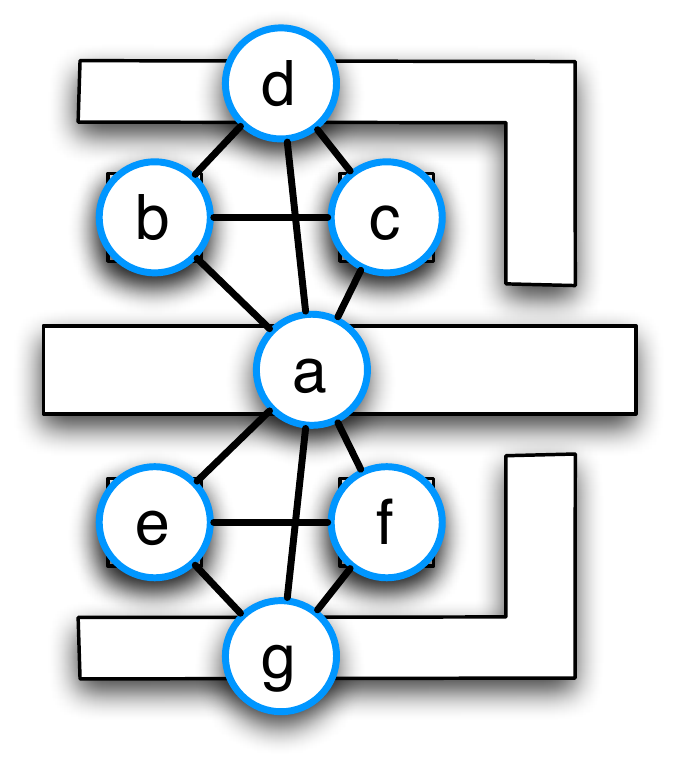}}
    \hspace{.1in}
    \subfigure[]{\includegraphics[width=0.16\textwidth]{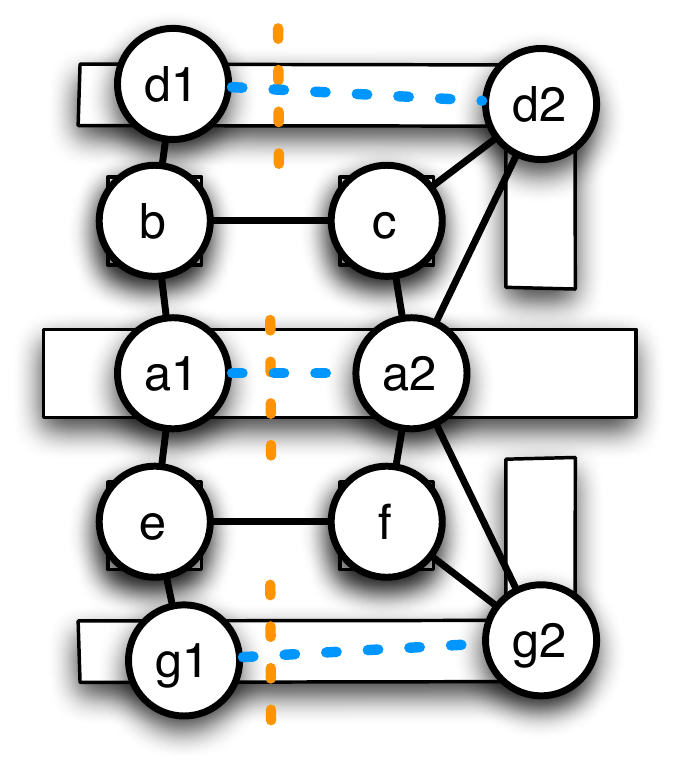}}
    %\hspace{-.1in}
    \subfigure[]{\includegraphics[width=0.16\textwidth]{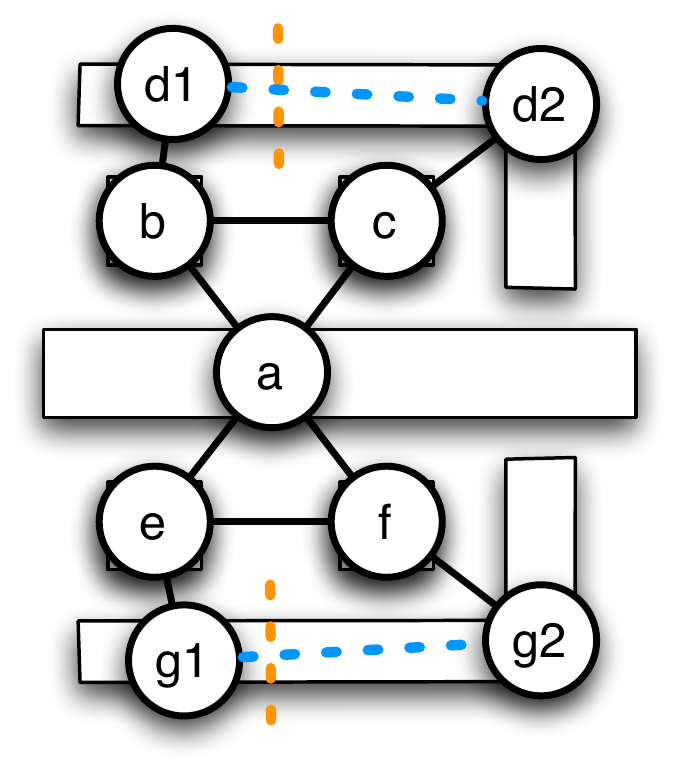}}
    \hspace{.1in}
    \subfigure[]{\includegraphics[width=0.15\textwidth]{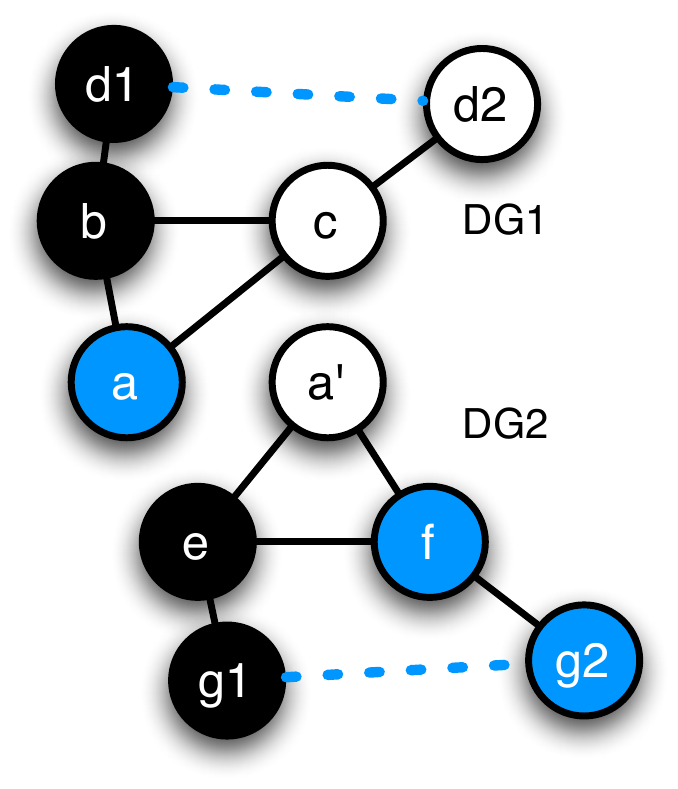}}
    %\hspace{-.6in}
    \vspace{-.1in}
    \caption{Layout graph cut vertex stitch forbiddance.}
    \label{fig:nostitch}
    \vspace{-.1in}
\end{figure}

A vertex of a graph is called a cut vertex if its removal decomposes the graph into two or more connected components.
Cut vertices can be identified through the process of bridge computation \cite{TPL_ICCAD2011_Yu}.
During stitch candidate generation, forbidding any stitch candidate on cut vertices can be helpful for later decomposition graph simplification.
Fig. \ref{fig:nostitch} (a) shows a layout graph, where feature $a$ is a cut vertex, since its removal can partition the layout graph into two parts: \{b, c, d\} and \{e, f, g\}.
If stitch candidates are introduced within $a$, the corresponding decomposition graph is illustrated in Fig. \ref{fig:nostitch} (b), which is hard to be further simplified.
If we forbid the stitch candidate on $a$, the corresponding decomposition graph is shown in Fig. \ref{fig:nostitch} (c), where $a$ is still cut vertex in decomposition graph.
Therefore we can apply 2-connected component computation \cite{TPL_DAC2012_Fang} to simplify the problem size, and apply color assignment separately (see Fig. \ref{fig:nostitch} (d)).

\vspace{-.05in}
\subsection{Decomposition Graph Vertex Clustering}
\label{sec:cluster}

\begin{figure}[tb]
    \centering
    \vspace{-.1in}
    \subfigure[]{\includegraphics[width=0.15\textwidth]{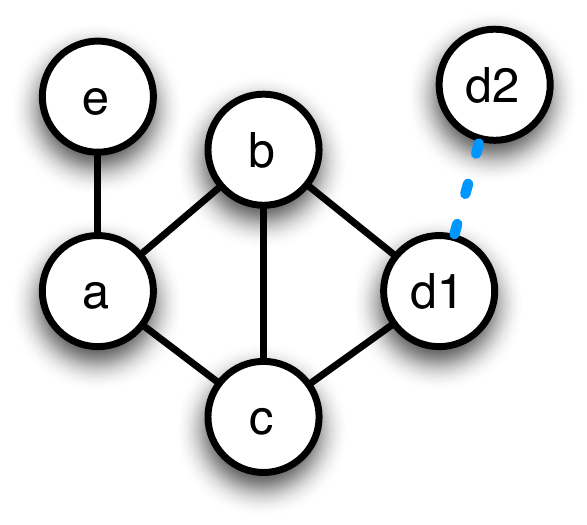}}
    %\hspace{-.1in}
    \subfigure[]{\includegraphics[width=0.15\textwidth]{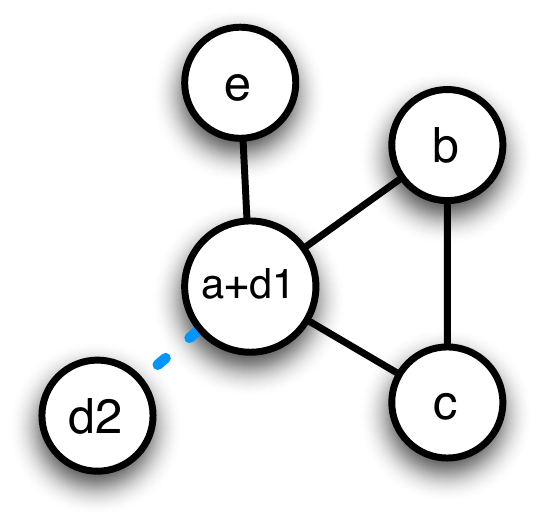}}
    \vspace{.1in}
    \caption{DG vertex clustering to reduce the decomposition graph size.}
    \label{fig:cluster}
    \vspace{-.1in}
\end{figure}

Decomposition graph vertex clustering is a speedup technique to further reduce the decomposition graph size.
As shown in Fig. \ref{fig:cluster} (a), vertices $a$ and $d_1$ share the same conflict relationships against $b$ and $c$.
Besides, there is no conflict edges between $a$ and $d_1$.
If no conflict is introduced, vertices $a$ and $d_1$ should be assigned the same color,
therefore we can cluster them together, as shown in Fig. \ref{fig:cluster} (b).
Note that the stitch and conflict relationships are also merged.
Applying vertex clustering in decomposition graph can further reduce the problem size.

\vspace{-.05in}
\subsection{Fast Color Assignment Trial}
\label{sec:fastassign}

Although the SDP and the partition based mapping can provide high performance for color assignment,
it is still expensive to be applied to all the decomposition graphs.
%To achieve further speed-up
We derive a fast color assignment trial before calling SDP based method.
If no conflict or stitch is introduced, our trial solves the color assignment problem in linear time.
%Otherwise, the comprehensive SDP based method is applied as discussed above.
Note that SDP method is skipped only when decomposition graph can be colored without stitch or conflict,
our fast trial does not lose any solution quality.
Besides, our preliminary results show that more than half of the decomposition graphs can be decomposed using this fast method.
Therefore, the runtime can be dramatically reduced.

\begin{algorithm}[htb]
\caption{Fast Color Assignment Trial} 
\label{alg:trial}
\begin{algorithmic}[1]
   \REQUIRE Decomposition graph $G$, stack $S$.
   \WHILE{$\exists n \in G$ s.t. $d_{conf}(n) < 3$ \& $d_{stit}(n) < 2$}
      \STATE $S.push(n)$; $G.delete(n)$;
   \ENDWHILE
   \IF{$G$ is not empty}
     \STATE Recover all vertices in $S$;
     \RETURN \FALSE;
     %\STATE Call SDP and partition based mapping;
   \ELSE
   \WHILE{ $ ! S.empty()$}
      \STATE $n = S.pop()$; $G.add(n)$;
	    \STATE Assign $n$ a legal color;
   \ENDWHILE
   \RETURN \TRUE;
   \ENDIF
\end{algorithmic}
\end{algorithm}

The fast color assignment trial is shown in Algorithm \ref{alg:trial}.
First, we iteratively remove the vertex with conflict degree ($d_{conf}$) less than 3 and stitch degree ($d_{stit}$) less than 2 (lines 1--3).
If some vertices cannot be removed, we recover all the vertices in stack $S$,
then return $false$;
%and apply SDP and partition based mapping (lines 4--6).
Otherwise, the vertices in $S$ are iteratively popped (recovered) (lines 8--12).
For each vertex $n$ popped, since it is connected with at most one stitch edge, we can always assign one color without introducing conflict or stitch.
%Through the layout graph simplification \cite{TPL_ICCAD2011_Yu}, for each vertex in the layout graph, its degree is larger than to 2.
%However, after introducing the stitch edges and the bridge computation, corresponding decomposition graphs would have some differences.

% Table 1: comparison on ISCAS85
%{{{
\begin{table*}[tb]
\centering
%\vspace{-.1in}
\renewcommand{\arraystretch}{0.9}
\caption{Comparison of Runtime and Performance.}
\label{tab:result1}
\begin{tabular}{|c|c|c|c|c||c|c|c|c||c|c|c|c||c|c|c|c|}
  \hline \hline
  \multirow{2}{*}{Circuit}&\multicolumn{4}{c||}{ICCAD'11 \cite{TPL_ICCAD2011_Yu}} &\multicolumn{4}{c||}{DAC'12 \cite{TPL_DAC2012_Fang}}
                          &\multicolumn{4}{c||}{DAC'13 \cite{TPL_DAC2013_Kuang}\footnotemark{}}  &\multicolumn{4}{c|}{SDP+PM}\\
  \cline{2-17} &cn\#  &st\#   & cost & CPU(s)       &cn\#  &st\#  &cost  &CPU(s)      &cn\#  &st\#  &cost  & CPU(s)     &cn\#  &st\#  &cost  & CPU(s) \\
  \hline                                                                                                                  
  C432         & 3    & 1     & 3.1  & 0.09         &0     & 6    & 0.6  & 0.03        &0    &4   &0.4    &0.01         &0     &4     &0.4   &0.2     \\
  C499         & 0    & 0     & 0    & 0.07         &0     & 0    & 0    & 0.04        &0    &0   &0      &0.01         &0     &0     &0     &0.2     \\
  C880         & 1    & 6     & 1.6  & 0.15         &1     & 15   & 2.5  & 0.05        &0    &7   &0.7    &0.01         &0     &7     &0.7   &0.3     \\
  C1355        & 1    & 2     & 1.2  & 0.07         &1     & 7    & 1.7  & 0.07        &0    &3   &0.3    &0.01         &0     &3     &0.3   &0.3     \\
  C1908        & 0    & 1     & 0.1  & 0.07         &1     & 0    & 1    & 0.1         &0    &1   &0.1    &0.01         &0     &1     &0.1   &0.3     \\
  C2670        & 2    & 4     & 2.4  & 0.17         &2     & 14   & 3.4  & 0.16        &0    &6   &0.6    &0.04         &0     &6     &0.6   &0.4     \\
  C3540        & 5    & 6     & 5.6  & 0.27         &2     & 15   & 3.5  & 0.2         &1    &8   &1.8    &0.05         &1     &8     &1.8   &0.5     \\
  C5315        & 7    & 7     & 7.7  & 0.3          &3     & 11   & 4.1  & 0.27        &0    &9   &0.9    &0.05         &0     &9     &0.9   &0.7     \\
  C6288        & 82   & 131   & 95.1 & 3.81         &19    & 341  & 53.1 & 0.3         &14   &191 &33.1   &0.25         &1     &213   &22.3  &2.7     \\
  C7552        & 12   & 15    & 13.5 & 0.77         &3     & 46   & 7.6  & 0.42        &1    &21  &3.1    &0.1          &0     &22    &2.2   &1.1     \\
  S1488        & 1    & 1     & 1.1  & 0.16         &0     & 4    & 0.4  & 0.08        &0    &2   &0.2    &0.01         &0     &2     &0.2   &0.3     \\
  S38417       & 44   & 55    & 49.5 & 18.8         &20    & 122  & 32.2 & 1.25        &19   &55  &24.5   &0.42         &19    &55    &24.5  &7.9     \\
  S35932       & 93   & 18    & 94.8 & 89.7         &46    & 103  & 56.3 & 4.3         &44   &41  &48.1   &0.82         &44    &48    &48.8  &21.4    \\
  S38584       & 63   & 122   & 75.2 & 92.1         &36    & 280  & 38.8 & 3.7         &36   &116 &47.6   &0.77         &37    &118   &48.8  &22.2    \\
  S15850       & 73   & 91    & 82.1 & 79.8         &36    & 201  & 56.1 & 3.7         &36   &97  &45.7   &0.76         &34    &101   &44.1  &20.0    \\
  \hline                                                                                                   
  avg.         &25.8  & 30.7  &28.9  &19.1          &11.3  &60.87 &17.42 & 0.978       &10.1 &37.4&13.8   &0.22         &9.07  &39.8  &13.0  &5.23     \\
  %\hline                                                                                                     
  ratio        &&&\textbf{2.2} &\textbf{3.65}       &&&\textbf{1.34}&\textbf{0.19}     &&&\textbf{1.06}&\textbf{0.04}   &&&\textbf{1.0}    &\textbf{1.0}       \\
  \hline \hline
\end{tabular}
\vspace{-.1in}
\end{table*}
%}}}

% Table 2: Comparison on very dense layout
%{{{
\begin{table*}[!bhtp]
\centering
%\vspace{-.1in}
\renewcommand{\arraystretch}{1.0}
\caption{Comparison on Very Dense Layouts}
\label{tab:result1}
\begin{tabular}{|c|c|c|c|c||c|c|c|c||c|c|c|c|}
  \hline \hline
  \multirow{2}{*}{Circuit}&\multicolumn{4}{c||}{ICCAD 2011 \cite{TPL_ICCAD2011_Yu}}&\multicolumn{4}{c||}{DAC 2012 \cite{TPL_DAC2012_Fang}}&\multicolumn{4}{c|}{SDP+PM}\\
  \cline{2-13}
             &cn\#     &st\#     & cost & CPU(s)            &cn\#     &st\#     &cost    &CPU(s)            &cn\#  &st\#    &cost    & CPU(s) \\
  \hline                                                                                                                
 mul\_top    &836      &44       &840.4   &236              &457      &0        &457     &0.8               &118   &271     &145.1      &57.6    \\
 exu\_ecc    &119      &1        &119.1   &11.1             &53       &0        &53      &0.7               &22    &64      &28.4       &4.3     \\
 c9\_total   &886      &228      &908.8   &47.4             &603      &641      &667.1   &0.52              &117   &1009    &217.9      &7.7     \\
 c10\_total  &2088     &554      &2143.4  &52               &1756     &1776     &1933.6  &1.1               &248   &1876    &435.6      &19      \\
 s2\_total   &2182     &390      &2221    &936.8            &1652     &5976     &2249.6  &4                 &703   &5226    &1225.6     &70.7    \\
 s3\_total   &6844     &72       &6851.2  &7510.1           &4731     &13853    &6116.3  &13.1              &958   &10572   &2015.2     &254.5   \\
 s4\_total   &NA       &NA       &NA      & $>$10000        &3868     &13632    &5231.2  &13                &1151  &11091   &2260.1     &306     \\
 s5\_total   &NA       &NA       &NA      & $>$10000        &4650     &16152    &6265.2  &12.9              &1391  &13683   &2759.3     &350.4   \\
  \hline                                                                                                                                        
  avg.       &NA       &NA       &NA      & $>$3600         &2221.3   &6503.8   &2871.6  &5.8               &588.5 &5474    &1135.9    &134     \\
  ratio      &&& -     &\textbf{$>$27.0}                    &&&\textbf{2.53}&\textbf{0.05}                  &&&\textbf{1.0} &\textbf{1.0}\\
  \hline \hline
\end{tabular}
\vspace{-.1in}
\end{table*}
%}}}

%\vspace{-.1in}
\section{Experimental Results}
\label{sec:result}

We implement our decomposer in C++ and test it on an Intel Xeon 3.0GHz Linux machine with 32G RAM.
%OpenAccess 2.2 \cite{OpenAccess} is adopted for interfacing with GDSII directly.
ISCAS 85\&89 benchmarks from \cite{TPL_ICCAD2011_Yu} are used,
where the minimum coloring spacing $dis_m$ was set the same with previous studies \cite{TPL_ICCAD2011_Yu}\cite{TPL_DAC2012_Fang}.
Besides, to perform a comprehensive comparison, we also test on other two benchmark suites.
The first suite is with six dense benchmarks (``c9\_total''-``s5\_total''),
while the second suite is two synthesized OpenSPARC T1 designs ``mul\_top'' and ``exu\_ecc'' with Nangate 45nm standard cell library \cite{nangate}.
%Cadence SOC Encounter \cite{socEncounter} is applied to perform placement and routing.
When processing these two benchmark suites we set the minimum coloring distance $dis_m = 2 \cdot w_{min}+3 \cdot s_{min}$,
where $w_{min}$ and $s_{min}$ denote the minimum wire width and the minimum spacing, respectively.
The parameter $\alpha$ is set as $0.1$.
The size of each bin is set as $10 \cdot dis_m \times 10 \cdot dis_m$.
We use CSDP \cite{CSDP} as the solver for the semidefinite programming (SDP).

\vspace{-.05in}
\subsection{Comparison with other decomposers}

\footnotetext{The results of DAC'13 decomposition are from \cite{TPL_DAC2013_Kuang}.}
In the first experiment, we compare our decomposer with the state-of-the-art layout decomposers which are not balanced density aware
\cite{TPL_ICCAD2011_Yu}\cite{TPL_DAC2012_Fang}\cite{TPL_DAC2013_Kuang}.
We obtain the binary files from \cite{TPL_ICCAD2011_Yu} and \cite{TPL_DAC2012_Fang}.
Since currently we cannot obtain the binary for decomposer in \cite{TPL_DAC2013_Kuang}, we directly use the results listed in \cite{TPL_DAC2013_Kuang}.
Here our decomposer is denoted as ``\textbf{SDP+PM}'', where ``PM'' means the partition based mapping.
The $\beta$ is set as 0.
In other words, SDP+PM only optimizes for stitch and conflict number.
Table \ref{tab:result1} shows the comparison in terms of runtime and performance.
For each decomposer we list its stitch number, conflict number, cost and runtime.
The columns ``cn\#" and ``st\#" denote the conflict number and the stitch number, respectively.
``cost'' is the cost function, which is set as cn\# $+ 0.1 \times$ st\#.
``CPU(s)" is computational time in seconds.

First, we compare SDP+PM with the decomposer in \cite{TPL_ICCAD2011_Yu}, which is based on SDP formulation as well.
%However, in their work stitch generation is directly follow previous DPL works and mapping is based on a greedy method.
From Table \ref{tab:result1} we can see that the new stitch candidate generation (see \cite{TPL_DAC2013_Kuang} for more details) and partition-based mapping can achieve better performance (reducing the cost by around 55\%).
Besides, SDP+PM can get nearly $4\times$ speed-up.
The reason is that, compared with \cite{TPL_ICCAD2011_Yu}, a set of speedup techniques, i.e.,
2-vertex-connected component computation, layout graph cut vertex stitch forbiddance (Sec. \ref{sec:nostitch}),
decomposition graph vertex clustering (Sec. \ref{sec:cluster}), and fast color assignment trial (Sec. \ref{sec:fastassign}),
 are proposed.
%Because SDP+PM may generate more stitch candidates, which increases the problem size, the runtime is longer.
%Considering the greatly improved solutions in terms of conflict and stitch number, this penalty is quite acceptable.
Second, we compare SDP+PM with the decomposer in \cite{TPL_DAC2012_Fang},
which applies several graph based simplifications and maximum independent set (MIS) based heuristic.
From Table \ref{tab:result1} we can see that although the decomposer in \cite{TPL_DAC2012_Fang} is
faster, MIS based heuristic has worse solution qualities (around 33\% cost penalty compared to SDP+PM).
Compared with the decomposer in \cite{TPL_DAC2013_Kuang}, although SDP+PM is slower, it can reduce the cost by around 6\%.
%Moreover, in \cite{TPL_DAC2012_Fang} the stitch minimization is applied as a post-process.
%Therefore, the stitch number is quite large, i.e., compared with ours it increases by 105\%.

In addition, we compare SDP-PM with other two decomposers \cite{TPL_ICCAD2011_Yu}\cite{TPL_DAC2012_Fang} for some very dense layouts, as shown in Table \ref{tab:result2}.
We can see that for some cases the decomposer in \cite{TPL_ICCAD2011_Yu} cannot finish in 1000 seconds.
Compared with \cite{TPL_DAC2012_Fang} work, SDP+PM can reduce cost by 65\%.
It is observed that compared with other decomposers, SDP+PM demonstrates much better performance when the input layout is dense.
The reason may be that when the input layout is dense, through graph simplification, each independent problem size may still be quite large,
then SDP based approximation can achieve better results than heuristic.
It can be observed that for the last three cases our decomposer could reduce thousands of conflicts.
Each conflict may require manual layout modification or high ECO efforts, which are very time consuming.
Therefore, even our runtime is more than \cite{TPL_DAC2012_Fang}, it is still acceptable (less than 6 minutes for the largest benchmark).

\vspace{-.05in}
\subsection{Comparison for Density Balance}

% table 3: for density balance (EPE)
%{{{
\begin{table}[bt]
\centering
%\vspace{-.2in}
\renewcommand{\arraystretch}{1.0}
\caption{Balanced density impact on EPE}
\label{tab:result2}
\begin{tabular}{|c|c|c|c||c|c|c|}
  \hline  \hline
  \multirow{2}{*}{Circuit} & \multicolumn{3}{c||}{SDP+PM} & \multicolumn{3}{c|}{SDP+PM+DB} \\
  \cline{2-7}
             &cost  & CPU(s) & EPE\#      &cost  & CPU(s) & EPE\#   \\
  \hline
  C432       &0.4   &0.2    &0            &0.4   &0.2    &0    \\
  C499       &0     &0.2    &0            &0     &0.2    &0    \\
  C880       &0.7   &0.3    &10           &0.7   &0.3    &7    \\
  C1355      &0.3   &0.3    &18           &0.3   &0.3    &15   \\
  C1908      &0.1   &0.3    &130          &0.1   &0.3    &58   \\
  C2670      &0.6   &0.4    &168          &0.6   &0.4    &105  \\
  C3540      &1.8   &0.5    &164          &1.8   &0.5    &79   \\
  C5315      &0.9   &0.7    &225          &1.0   &0.7    &115  \\
  C6288      &22.3  &2.7    &31           &32.0  &2.8    &15   \\
  C7552      &2.2   &1.1    &273          &2.5   &1.1    &184  \\
  S1488      &0.2   &0.3    &72           &0.2   &0.3    &44   \\
  S38417     &24.5  &7.9    &420          &24.5  &8.5    &412  \\
  S35932     &48.8  &21.4   &1342         &49.8  &24     &1247 \\
  S38584     &48.8  &22.2   &1332         &49.1  &23.7   &1290 \\
  S15850     &44.1  &20     &1149         &47.3  &21.3   &1030 \\
  \hline
  avg.       &13.0  &5.23   &355.6        &14.0  &5.64   &306.7\\
  ratio      &1.0   &1.0    &\textbf{1.0} &1.07  &1.08   &\textbf{0.86}\\
  \hline  \hline
\end{tabular}
\vspace{-.1in}
\end{table}
%}}}

In the second experiment, we test our decomposer for the density balancing.
We analyze edge placement error (EPE) using Calibre-Workbench \cite{Calibre} and industry-strength setup.
%
% get lithographic printed images for 45nm node test patterns.
%Our optical parameters are wavelength $\lambda = 193$nm, numerical aperture (NA) is set as 0.92 dry lithography,
%and dipole unpolarized illumination $\sigma = 0.9/0.7$.
%The thicknesses of the photo-resist (PR) and the bottom anti-reflecting coating (BARC) are 150nm and 38nm, respectively.
For analyzing the EPE in our test cases, we use systematic lithography process variation, such as focus $\pm 50$nm and dose $\pm 5\%$.
In Table \ref{tab:result2}, we compare SDP+PM with ``\textbf{SDP+PM+DB}'', which is our density balanced decomposer.
Here $\beta$ is set as 0.04 (we have tested different $\beta$ values, we found that bigger $\beta$ does not help much any more; meanwhile, we still want to give conflict and stitch higher weights).
Column ``cost'' also lists the weighted cost of conflict and stitch, i.e., cost $=$ cn\#$+ 0.1 \times$st\#.

From Table \ref{tab:result2} we can see that by integrating density balance into our decomposition flow,
our decomposer (SDP+PM+DB) can reduce EPE hotspot number by 14\%.
Besides, density balanced SDP based algorithm can maintain similar performance to the baseline SDP implementation:
only 7\% more cost of conflict and stitch, and only 8\% more runtime.
In other words, our decomposer can achieve a good density balance while keeping comparable conflicts/stitches.

% table for additional EPE compare
%{{{
\begin{table}[tb]
\centering
%\vspace{-.2in}
\renewcommand{\arraystretch}{1.0}
\caption{Additional Comparison for Density Balance}
\label{tab:denseEPE}
\begin{tabular}{|c|c|c|c||c|c|c|}
  \hline  \hline
  \multirow{2}{*}{Circuit} & \multicolumn{3}{c||}{SDP+PM} & \multicolumn{3}{c|}{SDP+PM+DB} \\
  \cline{2-7}
             &cost  & CPU(s) & EPE\#     &cost  & CPU(s) & EPE\#   \\
  \hline
 mul\_top    &145.1   &57.6   &632       &147.5     &63.8    &630      \\
 exu\_ecc    &28.4    &4.3    &140       &33.9      &4.8     &138      \\
 c9\_total   &217.9   &7.7    &60        &218.6     &8.3     &60       \\
 c10\_total  &435.6   &19     &77        &431.3     &19.6    &76       \\
 s2\_total   &1225.6  &70.7   &482       &1179.3    &75      &433      \\
 s3\_total   &2015.2  &254.5  &1563      &1937.5    &274.5   &1421     \\
 s4\_total   &2260.1  &306    &1476      &2176.3    &310     &1373     \\
 s5\_total   &2759.3  &350.4  &1270      &2673.9    &352     &1171     \\
  \hline
  avg.       &1135.9  &134    &712.5        &1099.8    &138.5   & 662.8   \\
  ratio      &1.0     &1.0    &\textbf{1.0} &0.97      &1.04    & \textbf{0.93}\\
  \hline  \hline
\end{tabular}
\vspace{-.1in}
\end{table}
%}}}

We further compare the density balance, especially EPE distributions for very dense layouts.
As shown in Table \ref{tab:denseEPE}, our density balanced decomposer (SDP+PM+DB) can reduce EPE distribution number by 7\%.
Besides, for very dense layouts, density balanced SDP approximation can maintain similar performance with plain SDP implementation: only 4\% more runtime.

\vspace{-.05in}
\subsection{Scalability of SDP Formulation}

\begin{figure}[tbh]
  \vspace{-.1in}
  \centering
  \includegraphics[width=0.44\textwidth]{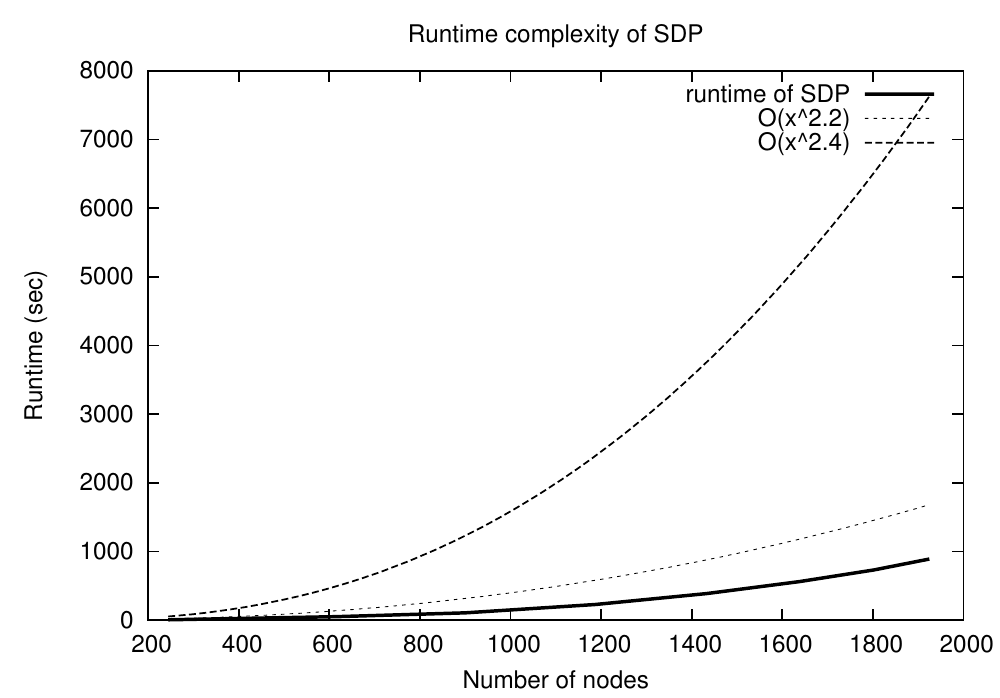}
  \caption{~Scalability of SDP Formulation.}
  \label{fig:scalability}
  \vspace{-.1in}
\end{figure}

In addition, we demonstrate the scalability of our decomposer, especially the SDP formulation.
Penrose benchmarks from \cite{TPL_SPIE08_Cork} are used to explore the scalability of SDP runtime.
No graph simplification is applied, therefore all runtime is consumed by solving SDP formulation.
Fig. \ref{fig:scalability} illustrates the relationship between graph (problem) size against SDP runtime.
Here the X axis denotes the number of nodes (e.g., the problem size), and the Y axis shows the runtime.
We can see that the runtime complexity of SDP is less than O($n^{2.2}$).

\vspace{-.1in}
\section{Conclusion}
\label{sec:conclusion}

In this paper, we propose a high performance TPL layout decomposer with balanced density.
%We carefully analyze some limitations of previous works, and propose a set of algorithms,
Density balancing is integrated into all the key steps of our decomposition flow.
In addition, we propose a set of speedup techniques, such as layout graph cut vertex stitch forbiddance,
decomposition graph vertex clustering, and fast color assignment trial.
Compared with state-of-the-art frameworks, our decomposer demonstrates the best performance in minimizing the cost of conflicts and stitches.
%can achieve much better performance.
%(31\%--54\% cost reduction).
Furthermore, our balanced decomposer can obtain less EPE while maintaining very comparable conflict  and stitch results.
As TPL may be adopted by industry for 14nm/11nm nodes, we believe more research will be needed to enable TPL-friendly design and mask synthesis.

\vspace{-.1in}
\section*{Acknowledgment}

This work is supported in part by NSF grants CCF-0644316 and CCF-1218906, SRC task 2414.001, NSFC grant 61128010, and IBM Scholarship.
\vspace{-.1in}

{
%\vspace{-.1in}
%\scriptsize
\bibliographystyle{IEEEtran}
\bibliography{./Ref/Bei,./Ref/MPL,./Ref/Algorithm,./Ref/Partition,./Ref/EBL,./Ref/Lith,./Ref/Floorplan}
}

%\appendix
%\input{doc/app_sdp}
%\input{doc/app_stitch}
%\input{doc/app_proof}

\end{document}